\documentclass[a4paper]{article}
\usepackage[margin=25mm]{geometry}
\usepackage{amsmath}
\usepackage{amsfonts}
\usepackage{amssymb}
\usepackage{amsthm}
\usepackage{graphicx}
\usepackage{booktabs,subfig,longtable}
\usepackage{caption}
\usepackage{verbatim}
\usepackage[sort]{natbib}
\usepackage{multirow}
\usepackage{floatrow}
\immediate\write18{texcount -tex -sum  \jobname.tex > \jobname.wordcount.tex}

\usepackage[utf8]{inputenc}

\graphicspath{{plots/}}
\DeclareMathOperator*{\argmin}{arg\,min}

\makeatletter
\newcommand*{\rom}[1]{\expandafter\@slowromancap\romannumeral #1@}
\makeatother

\newtheorem{theorem}{Theorem}[section]
\newtheorem{lemma}[theorem]{Lemma}
\newtheorem{corollary}[theorem]{Corollary}
\newtheorem{proposition}[theorem]{Proposition}

\theoremstyle{remark}

\theoremstyle{plain}

\theoremstyle{definition}
\newtheorem{remark}{{\bf Remark}}
\newtheorem{assumption}{{\bf Assumption}}

\newcommand{\E}{\mathit{E}}

\def\var{{\mathrm{var}}}

\newcommand{\tr}{\mathrm{trace}}
\newcommand{\indep}{\;\, \rule[0em]{.03em}{.6em} \hspace{-.25em}
	\rule[0em]{.65em}{.03em} \hspace{-.25em}
	\rule[0em]{.03em}{.6em}\;\,}

\newcommand{\trans}{^{\mbox{\tiny {\sf T}}}}

\def\R{\mathbb R}

\def\nano{\scriptscriptstyle}
\newcommand\hi[1]{^{\nano #1}}
\def\real{\mathbb R}
\newcommand\ca[1]{{\cal{#1}}}
\newcommand\lo[1]{_{\nano #1}}

\def\var{\mathrm{var}}
\def\loo#1{\lo {\mathrm{#1}}}

\def\ran{\mathrm{ran}}
\def\ker{\mathrm{ker}}
\def\cran{\overline{\mathrm{ran}}}

\def\cov{\mathrm{cov}}
\def\nano{\scriptscriptstyle}
\def\inv{\hi{\nano -1}}

\def\nano{\scriptscriptstyle}

\def\ka{\kappa}
\def\cspan{\overline{\mathrm{span}}}
\def\tsum{\textstyle{\sum}}
\def\ali{&\,}

\def\mpinv{\hi{\dagger}}

\def\mpinv{\hi{\dagger}}

\def\convd{\stackrel{d}{\longrightarrow} }

\newfont{\rsfsten}{rsfs10 scaled 1050}
\newfont{\rsfstena}{rsfs10 scaled 750}
\newfont{\rsfstenb}{rsfs10 scaled 800}

\usepackage{titlesec}

\titleformat{\section}
{\normalfont\large\bfseries}{\thesection}{1em}{}
\titleformat{\subsection}
{\normalfont\normalsize\bfseries}{\thesubsection}{1em}{}
\titleformat{\subsubsection}
{\normalfont\normalsize\bfseries}{\thesubsubsection}{1em}{}
\titleformat{\paragraph}[runin]
{\normalfont\normalsize\bfseries}{\theparagraph}{1em}{}
\titleformat{\subparagraph}[runin]
{\normalfont\normalsize\bfseries}{\thesubparagraph}{1em}{}

\def\E{{\rm E}\,}

\include{apalike}

\def\nat{\mathbb{N}}
\def\pri{\hi{\prime}}

\providecommand{\keywords}[1]
{
	\small	
	\textbf{\textit{Keywords---}} #1
}

\title{Nonlinear function-on-function regression by RKHS}
\author{Peijun Sang$^{1}$, Bing Li$^{2}$  \\
	\small $^{1}$Department of Statistics and Actuarial Science,
	University of Waterloo, \\
	\small $^{2}$Department of Statistics, 
	Pennsylvania State University \\
}
\date{} 

\begin{document}
	\maketitle
	
	\begin{abstract}
	We propose a nonlinear function-on-function regression model where both the covariate and the response are random functions. The nonlinear regression is carried out in two steps: we first construct Hilbert spaces to accommodate the functional covariate and the functional response, and then build a second-layer Hilbert space for the covariate to capture nonlinearity. The second-layer space is assumed to be a reproducing kernel Hilbert space, which is generated by a positive definite kernel determined by the inner product of the first-layer Hilbert space for $X$--this structure is known as the nested Hilbert spaces. We develop estimation procedures to implement the proposed method, which allows the functional data to be observed at different time points for different subjects. Furthermore, we establish the convergence rate of our estimator as well as the 
	weak convergence of the predicted response in the Hilbert space. Numerical studies including both simulations and a data application are conducted to investigate the performance of our estimator in finite sample. 
	\end{abstract} \hspace{10pt}
	
	\keywords{functional data analysis, linear operator, Tikhonov regularization, weak convergence, Hawaii Ocean Time-series}
	
\section{Introduction} \label{sec-intro}

With the development of techniques in data collection, functional data have become increasingly common in modern statistical applications. As a useful tool to treat such data, functional data analysis (FDA) has been widely applied to diverse fields such as neural science, chemometrics, environmetrics and finance. A comprehensive overview of FDA can be found in several monographs (\citealp{ramsay2005}; \citealp{ferraty2006}; \citealp{kokoszka2017}). An important problem in FDA is to study the relationship between a response, which can be a scalar, a vector or a function, and a functional covariate. The aforementioned references as well as other monographs like \cite{horvath2012} and \cite{hsing2015} describe many ideas and methods to tackle this problem. In this paper, we propose a general regression model that allows for flexible nonlinear relations between the functional covariate and the functional response. 

The scalar-on-function regression, where a scalar response is regressed against a functional covariate, has been extensively studied. In particular, estimation and inference for linear scalar-on-function regression models have been one of the focal points of the FDA research in the past two decades. There are two mainstream approaches to fitting a functional linear model. The first represents the slope function in the linear model as a linear combination of a finite number of basis functions, so that fitting the model is reduced to estimating the linear coefficients. The basis can be either a pre-determined basis such as the B-spline basis or a data-driven basis such as the estimated eigenfunctions from the functional principal component analysis (FPCA). See \cite{cardot2003}, \cite{hall2007} and references therein. The second approach assumes that the slope function belongs to a reproducing kernel Hilbert space (RKHS) and resorts to the representer theorem to fit the model; see, for example, \cite{yuan2010} and \cite{cai2012}. In the asymptotic development, the aforementioned work mainly considered convergence rates of estimation and prediction in functional linear models. \cite{shang2015} studied confidence intervals for regression mean and prediction intervals for a future response in generalized functional linear models. \cite{cuesta2019} considered goodness-of-fit of functional linear models to investigate whether they adequately characterize the relation between a scalar response and a functional covariate. 

Unlike scalar-on-function regression, function-on-function regression has not yet been extensively developed due to its computational complexity. However, predicting a random function by another random function is an important problem in many applications. For instance, in the Canadian weather data set (\citealp{ramsay2005}, Chapter 12.4), the temperature profile is used to predict the annual profile of precipitation, rather than the total precipitation. As with scalar-on-function regression, linear models were most frequently investigated in the literature of function-on-function regression. \cite{ramsay2005} used tensor products of spline basis functions  and \cite{yao2005b} and \cite{crambes2013} used estimated eigenfunctions through the FPCA  to estimate the (bivariate) slope function in a linear function-on-function regression model. The idea of a regularized estimation through RKHS proposed by \cite{yuan2010} for a linear scalar-on-function regression model was extended to the function-on-function case by \cite{sun2018optimal}. However, as pointed out by \cite{muller2008functional}, linear relations may not be adequate to characterize dependence of one random function on another function in some applications. The Hawaii ocean data set studied in \cite{qi2019} justified this statement. Nonlinear function-on-function regression models have received substantially less attention in FDA (\citealp{reimherr2018}), let alone statistical inference of such models. 

In this paper, we propose a nonlinear function-on-function regression model that relies on the structure of the nested Hilbert spaces. Specifically, we assume that the functional covariate $X$ and the functional response $Y$ each resides in a Hilbert space, denoted by $\ca H \lo X$ and $\ca H \lo Y$, respectively. To achieve a nonlinear relation, we build another Hilbert space of functions on $\ca H \lo X$, which is assumed to be an RKHS generated by a positive definite kernel $\ka$ determined by the inner product of $\ca H \lo X$. Following \cite{li2017}, we refer to this space as the second-layer space, and the two-space structure as the nested Hilbert spaces. Lastly, we construct a linear operator from the second-layer RKHS to the target space $\ca H \lo Y$, which gives rise to a nonlinear relation. This idea is similar in spirit to that of support vector machines for regression (\citealp{friedman2001elements}, Chapter 5.8). We then propose an implementation method based on this nonlinear function-on-function regression model, which allows for profiles with irregularly spaced observed time points for both the covariate and the response. Numerical studies demonstrate that our proposed method can still achieve relatively good predictive performance when random functions are sparsely observed. We establish consistency with the convergence rate of our estimator under some mild conditions. We find that the convergence rate of the estimated regression operator can be improved from $n\hi{-1/4}$ in \cite{li2017} in function-on-function sufficient dimension reduction to $n\hi{-1/3}$ in the current setting under mild conditions.
More importantly, unlike the previous theoretical work on function-on-function regression that was focused on convergence rates without an asymptotic distribution, we establish weak convergence of the predicted mean for a future observation in $\ca H \lo Y$. This enables us to construct both pointwise confidence intervals and a simultaneous confidence band for the conditional mean. 

The rest of the paper is organized as follows. In Section \ref{sec-model}, we propose the nonlinear function-on-function regression model. In Section \ref{sec-estimation}, we develop an algorithm to fit the model and propose suitable methods to select the tuning parameters that are involved in the estimation procedure. Consistency and weak convergence of the estimator proposed in Section \ref{sec-estimation} are studied in Sections \ref{section:convergence rate} and \ref{sec-CLT}. Pointwise confidence intervals and a simultaneous confidence band for the conditional mean are then constructed based on the weak convergence result. In Section \ref{sec-simu}, we conduct simulation studies to investigate the performance of our proposed model in finite samples. The new model is applied to a data set to further demonstrate its performance in Section \ref{sec-real}.
Some concluding remarks are made in Section \ref{sec-conc}. All technical proofs are delegated to the supplementary material.

\section{Model construction} \label{sec-model}

In this section, we first introduce the concept of the nested Hilbert space, which plays an important role in our model construction, and then lay out the detailed steps to build our nonlinear function-on-function regression model. 

\subsection{Nested Hilbert space for predictor} \label{sec-nested}

Let $(\Omega, \ca F, P)$ be a probability space, $T$ an interval in $\real$, $\ca H \lo X$ and $\ca H \lo Y$   Hilbert spaces of functions on $T$. Let $X: \Omega \to \ca H \lo X$, $Y: \Omega \to \ca H \lo Y$ be random elements in $\ca H \lo X$ and $\ca H \lo Y$ measurable with respect to $\ca F / \ca F \lo X$ and $\ca F / \ca F \lo Y$, where $\ca F \lo X$ and $\ca F \lo Y$ denote the Borel $\sigma$-algebra generated by the open sets in $\ca H \lo X$ and $\ca H \lo Y$.  Let $P \lo X$ and $P \lo Y$  denote the distributions of $X$ and $Y$, and $P \lo {Y | X}: \ca H \lo X \times \ca F \lo Y \rightarrow \R$  the conditional distribution of $Y$ given $X$. 

Let $\ka: \ca H \lo X \times \ca H \lo X \to \real$ be a positive definite kernel and $\frak M \lo X$ be the RKHS generated by $\ka$.  We assume that $\ka$ is induced by the inner product in $\ca H \lo X$; that is, there exists a function $\rho : \R\hi 3 \rightarrow \R\hi +$, such that for any $f, g \in \ca H \lo X$, 
$$
\ka(f, g) = \rho(\langle f, f \rangle\lo {\ca H \lo X},  \langle f, g \rangle\lo {\ca H \lo X},    \langle g,  g \rangle\lo {\ca H \lo X}).
$$
An example of such a kernel is $\ka (f, g) = \exp (-\gamma \| f - g \|\lo{\ca H \lo X}\hi 2)$, where $\gamma > 0$ is a tuning constant.  This is an extension of the Gaussian radial basis function with the Euclidean norm replaced by the $\ca H \lo X$-norm. Since the kernel of $\frak M \lo X$ is determined by the inner product of $\ca H \lo X$, we refer to the $\frak M \lo X$ the nested RKHS via $\rho$ (\citealp{li2017}).

\subsection{Nonlinear function-on-function regression} 

Let $L\lo 2(P\lo X)$ denote the class of all measurable functions of $X$ such that 
$\E[f\hi 2(X)] < \infty$ under $P \lo X$. Let $ L\lo 2(P\lo Y)$ be defined in the same way for $Y$. In the following, $\frak M \lo X + \R$ represents the space $\{f + c: f \in \frak M \lo X, c \in \R\}$. 

\begin{assumption} \label{ass-dense}
	$\frak M \lo X + \R$ is a dense subset of $L\lo 2(P\lo X)$
	and $ \E (\|Y\|\lo {\ca H \lo Y }\hi 2) < \infty $.
\end{assumption}

Assumption \ref{ass-dense} is essentially the same as Assumption (AS) in \cite{fukumizu2009}. This assumption ensures that any function $f \in L\lo 2(P \lo X)$ can be approximated by a sequence of functions $\{f\lo n\} \subseteq \frak M \lo X$ in the sense that $\var(f - f\lo n) \rightarrow 0$. For two Hilbert spaces $\ca H\lo 1$ and $\ca H\lo 2$,  let $\ca B(\ca H\lo 1, \ca H\lo 2)$ denote the class of all bounded linear operators from $\ca H\lo 1$ to $\ca H\lo 2$; the special case $\ca B(\ca H, \ca H)$ is abbreviated by $\ca B (\ca H)$. For a linear operator $A \in \ca B (\ca H\lo 1, \ca H\lo 2)$, we use $\ker(A)$ to denote the kernel of $A$; that is, $\ker(A) = \{f \in \ca H \lo 1: Af = 0\}$; we use $\ran(A)$ to denote the range of $A$; that is, $\ran(A) = \{Af: f \in \ca H\lo 1\}$; we use $\cran(A)$ to denote the closure of $\ran(A)$ in $\ca H\lo 2$. For a self-adjoint operator $A \in \ca B(\ca H)$, we have 
$\ker(A)\hi{\perp} = \cran(A)$ and $\cran(A)\hi{\perp} = \ran(A)\hi{\perp} = \ker(A)$. 

\begin{assumption}
	There exists a constant $C > 0$ such that for any $f \in \frak M \lo X$, $\E[f\hi 2(X)] \leq C \|f\|\lo { \frak M \lo X}\hi 2$. 
	\label{ass-var}
\end{assumption}

Assumption \ref{ass-var} ensures that the inclusion mapping $\frak M \lo X \rightarrow L\lo 2(P\lo X)$, $f \mapsto f$ is a bounded linear operator. 
It also guarantees that the bilinear form $\frak M \lo X \times \frak M \lo X \rightarrow \R$, $(f, g) \mapsto \cov(f(X), g(X) )$ is bounded. Therefore, there exists an operator 
$\Sigma \lo {XX}  \in \ca B (\frak M \lo X)$ such that $\langle f, \Sigma \lo {XX}  g \rangle\lo{\frak M \lo X } = \cov(f(X), g(X) )$. Similarly, under Assumptions \ref{ass-dense} and \ref{ass-var}, the bilinear form $\frak M \lo X \times \ca H \lo Y  \rightarrow \R$, $(f, g) \mapsto \cov(f(X), \langle g, Y \rangle \lo {\ca H \lo Y} )$ is bounded, and this implies that there is
an operator $\Sigma \lo {XY}  \in \ca B (\ca H \lo Y, \frak M \lo X)$ such that $\langle f, \Sigma \lo {XY}  g \rangle \lo {\frak M \lo X } =  \cov(f(X), \langle g, Y \rangle \lo {\ca H \lo Y})$
for any $f \in \frak M \lo X$ and $g \in \ca H \lo Y$.  Moreover, Assumption \ref{ass-var} also implies that the linear functional $f \mapsto \E[f(X)]$ on $\frak M \lo X$ is bounded. Let $\mu\lo X$ denote the Riesz representation of this linear functional; that is, $\langle f, \mu\lo X \rangle \lo {\frak M \lo X} = \E[f(X)]$  for all $f \in \frak M \lo X$. By construction, for any $x \in \ca H \lo X$, $\mu\lo X(x) = \langle \mu\lo X, \ka(\cdot, x) \rangle \lo {\frak M \lo X} = \E[\ka (X, x)]$.
Similarly, $\mu\lo Y$ is defined as the Riesz representation of the bounded linear functional $f \mapsto \E[\langle Y, f\rangle \lo {\ca H \lo Y}]$ where $f \in \ca H \lo Y$, that is, $\langle f, \mu\lo Y \rangle \lo {\ca H \lo Y} = \E[\langle f, Y \rangle \lo {\ca H \lo Y} ]$ for $f \in \ca H \lo Y$. The function $\mu\lo Y$ is the first moment of $Y$, and is denoted by $\E(Y)$; the function $\mu\lo X$ is called the mean element of $X$ in $\frak M \lo X$.

Using the same argument as in \cite{fukumizu2009}, it can be shown that
\begin{align*}
	& \Sigma \lo {XX}  = \E\{(\ka (\cdot, X) - \mu \lo X) \otimes (\ka (\cdot, X) - \mu \lo X)\},~~
	\Sigma \lo {XY} = \E\{(\ka (\cdot, X) - \mu \lo X) \otimes (Y - \mu\lo Y) \},\\
	& \Sigma\lo {YX} = \Sigma\lo {XY}\hi *  = \E\{(Y - \mu\lo Y)  \otimes (\ka (\cdot, X) - \mu \lo X)\}, 
\end{align*}
under Assumption \ref{ass-var}. Since $\cran(\Sigma\lo {XX})\hi\perp = \ker(\Sigma \lo {XX})$ consists of functions that are constants almost surely, $\cran(\Sigma\lo {XX})$ is the effective domain of $\Sigma\lo {XX}$. As shown in \cite{li2017}, this effective domain can be represented explicitly using the kernel as
\begin{equation} \label{eq-ranX}
	\cspan\{\ka(\cdot, x) - \mu\lo X : x \in \ca H \lo X\}.
\end{equation}
We will use $ \frak M \lo X\hi 0$ to denote the effective domain.

For any $g \in \ran(\Sigma \lo {XX})$, there exists some $f \in \frak M \lo X$ such that $\Sigma \lo {XX} f = g$. By Theorem 3.3.7 of \cite{hsing2015} and the fact that $\cran (\Sigma \lo {XX}) = \frak M \lo X\hi 0$, there exists a unique decomposition $f = f\lo 1 + f\lo 2$ such that $f\lo 1 \in \ker(\Sigma \lo {XX} )$ and $f\lo 2 \in \cran (\Sigma \lo {XX})  = \frak M \lo X\hi 0$. Therefore the mapping $g \mapsto f\lo 2$ from $\ran(\Sigma\lo {XX})$ to $\cran(\Sigma\lo {XX})$ is well-defined under Assumption \ref{ass-var}. We call this mapping the Moore-Penrose inverse of $\Sigma \lo {XX}$, and denote it by $\Sigma \lo {XX} \hi {\mpinv}$.

\begin{assumption} \label{ass-bounded}
	$\ran(\Sigma \lo {XY} ) \subseteq \ran(\Sigma \lo {XX})$ and $\Sigma \lo {XX}  \mpinv \Sigma \lo {XY}$ is a bounded operator. 
\end{assumption}

Under this assumption, $\Sigma \lo {XX} \mpinv\Sigma \lo {XY}$ is a well-defined bounded operator. Generally speaking, $\Sigma \lo {XX} \mpinv$ is not bounded since 
$\Sigma \lo {XX} $ is a Hilbert-Schmidt operator (\citealp{fukumizu2009}). However, as argued in \cite{li2017}, it is reasonable to assume that
$\Sigma \lo {XX} \mpinv \Sigma \lo {XY}$ is bounded, which is determined by the interaction of these two operators. 

Our function-on-function regression problem is to find $B \lo 0 \in \ca B ( \ca H \lo Y, \frak M \lo X \hi 0)$ such that, for each $g \in \ca H \lo Y $, 
\begin{align}\label{eq:ols}
	B\lo 0 = \argmin \lo {B \in \ca B ( \ca H \lo Y , \frak M \lo X \hi 0)}\E \left[ \left\{ \langle g, Y \rangle \lo {\ca H \lo Y} - \E [\langle g, Y \rangle \lo {\ca H \lo Y} ] - [(B g)(X) - \E\{(Bg)(X)\}]  \right\} \hi 2 \right].
\end{align}
This is, indeed, a generalization of multivariate regression. 


\begin{theorem}  \label{th-sol}
	Under Assumptions 1-3, the solution to \eqref{eq:ols} is
	\begin{align*}
		B \lo 0 = \Sigma \lo {XX} \mpinv \Sigma \lo {XY}
	\end{align*}
\end{theorem}

The operator $B\lo 0$ is a special case of the ``regression operator'' defined in \cite{lee2016variable}, and was also used implicitly in \citet{fukumizu2004, fukumizu2009}.
We now develop its properties in the context of nonlinear function-on-function regression, which is important for later development. The next proposition describes a relation between $B\lo 0$ and the conditional expectation $\E[\langle g, Y \rangle \lo {\ca H \lo Y} \mid X]$ for any $g \in \ca H \lo Y$. 

\begin{proposition}  \label{prop-reg}
	Under Assumptions 1-3, we have, for any $g \in \ca H \lo Y $, 
	\begin{align} \label{eq-reg}
		\E[\langle g, Y \rangle \lo {\ca H \lo Y} | X] = (B\lo 0 g)(X) + \E[\langle g, Y \rangle \lo {\ca H \lo Y}] - \E[(B\lo 0 g)(X)].
	\end{align}
\end{proposition}

We define a (random) linear functional $T\lo X: \ca H\lo Y \rightarrow \real, h \mapsto 
\E[\langle Y, h\rangle \lo {\ca H\lo Y} | X]$. By construction, $\E(Y | X)$ is the Riesz representation of $T\lo X$. Proposition \ref{prop-reg} leads to the following relation between $B\lo 0$ and $\E(Y | X)$, which is a random element in $\ca H \lo Y$.

\begin{proposition}  \label{prop-condE}
	Under Assumptions 1-3, 
	\begin{align} \label{eq-condE}
		\E(Y | X) = B\lo 0 \hi * [\ka (\cdot, X) - \mu\lo X] + \mu\lo Y,
	\end{align}
	where $B\lo 0 \hi * = \Sigma \lo {YX}  \Sigma \lo {XX} \mpinv \in \ca B(\frak M \hi 0\lo X, \ca H \lo {Y} )$ is the adjoint operator of $B\lo 0$. 
\end{proposition}

For convenience, in the following, we use $Y\lo c$ to represent $Y - \mu\lo Y$ and use $\ka\lo c(\cdot, x)$ to represent $\ka(\cdot, x) - \mu\lo X$. Note that, for any $f \in \frak M \lo X$, we have $\langle \ka\lo c(\cdot, x), f \rangle \lo {\frak M \lo X} = f(x) - \E[f(X)]$. 
Proposition \ref{prop-condE} indicates that, for a given $x \in \ca H\lo X$, the predicted value of $Y$  is given by
\begin{align}
	\nonumber
	\E(Y | X=x) & = \Sigma\lo {YX}\Sigma\lo {XX}\mpinv \ka\lo c (\cdot, x)  + \mu\lo Y \\
	& = \E\left[\left\{(\Sigma \lo {XX} \mpinv \ka\lo c(\cdot, x))(X)\right\}Y\lo c \right] + \mu\lo Y.
	\label{eq-condexp}
\end{align}

\section{Estimation}
\label{sec-estimation}
In the last section we have described the solution to the nonlinear function-on-function regression at the population level. In this section, we implement the regression at the sample level.  The key step is to construct the sample estimate of the regression operator based on $n$ i.i.d. observations on $(X, Y)$ by representing relevant operators as $n \times n$ matrices with a coordinate representation system. 
See, for example, \cite{johnson1985} and \cite{li2018}.

\subsection{Coordinate representation system} \label{sec-coor}
Suppose that $\ca L\lo 1$ is a finite-dimensional linear space with basis $\ca B = \{\xi\lo 1, \ldots, \xi\lo p\}$. Then for any $\xi \in \ca L\lo 1$, there is a unique vector $(a\lo 1, \ldots, a\lo p)\trans \in \R\hi p$ such that $\xi = \sum\lo {i = 1}\hi p a\lo i \xi\lo i$. The vector $(a\lo 1, \ldots, a\lo p)\trans$ is called the coordinate of $\xi$ with respect to $\ca B$, and denoted as $[\xi] \lo {\ca B}$.  Throughout this section we will reserve the square brackets $[\cdot]$ exclusively for coordinate representation. Next, we introduce the coordinate representation of a linear operator between two (finite-dimensional) linear spaces. Suppose $\ca L\lo 2$ is another linear space with basis $\ca C = \{\eta\lo 1, \ldots, \eta\lo q\}$ and
$A$ is a linear operator from $\ca L\lo 1$ to $\ca L\lo 2$. Then for any $\xi \in \ca L\lo 1$, we have
$$
A \xi = A\left(\sum\lo {i = 1}\hi p ([\xi] \lo {\ca B})\lo i \xi\lo i\right) = \sum\lo{i = 1}\hi p ([\xi] \lo {\ca B})\lo i  (A\xi\lo i) = \sum\lo {i = 1}\hi p ([\xi] \lo {\ca B})\lo i  \sum\lo {j = 1}\hi q ([A\xi\lo i] \lo {\ca C})\lo j \eta\lo j. 
$$
By the law of matrix multiplication, we can rewrite the right-hand side of the above equation as
$$
\sum\lo {j = 1}\hi q \sum\lo {i = 1}\hi p ([A\xi\lo i] \lo {\ca C})\lo j ([\xi] \lo {\ca B})\lo i  \eta\lo j = \sum\lo {j = 1}\hi q \{({}\lo {\ca C}[A]\lo {\ca B})([\xi] \lo {\ca B})\}\lo j \eta\lo j, 
$$
where ${{}\lo{\ca C}[A] \lo {\ca B}}$ is the $q \times p$ matrix with $(i, j)$th entry being $([A\xi\lo j] \lo {\ca C})\lo i$.  This equation indicates that $[A \xi] \lo {\ca C} =( {}\lo{\ca C}[A]\lo {\ca B})([\xi] \lo {\ca B})$. We therefore call the matrix ${}\lo{\ca C}[A]\lo{\ca B}$ the coordinate of the linear operator $A$ with respect to bases $\ca B$ and $\ca C$. If we have a third linear space $\ca L\lo 3$ with basis $\ca D = \{\zeta\lo 1, \ldots, \zeta\lo l\}$ and another linear operator $B: \ca L\lo 2 \rightarrow \ca L\lo 3$, then it is straightforward to show that
${{}\lo {\ca D}[BA] \lo {\ca B} }= ({{}\lo{\ca D}[B]\lo{\ca C}}) ({{}\lo{\ca C}[A]\lo{\ca B}})$. When the relevant bases are clear from the context and no confusion will be caused, we will drop subscripts and write ${{}\lo{\ca C}[A]\lo{\ca B}}$ and $[\xi] \lo {\ca B}$ as $[A]$ and $[\xi]$, respectively. 

\subsection{Construction of $\ca H \lo X$, $\ca H \lo Y$ and $\frak M \lo X$}
\label{sec-const}

Let $(X\lo 1, Y\lo 1), \ldots, (X\lo n, Y\lo n)$ be i.i.d. observations of $(X, Y)$. 
In practice, instead of observing the whole trajectory of $X\lo i$, we have observations at only a finite subset of $T$. Let $\{t\lo {i1}, \ldots, t \lo {im\lo i} \}$ be the set of time points at which $X\lo i$ is observed, which may vary from subject to subject. Let
$$
V = \cup\lo {i = 1}\hi n \{t\lo {i1}, \ldots, t \lo {im\lo i}\}.
$$
Let $m$ be the cardinality of $V$, and $\nu\lo 1, \ldots, \nu\lo m$ the (relabeled) members of $V$. Note that $m$ may or may not be $\sum\lo {i = 1}\hi n m\lo i$. 
Let $J\lo i$ be the index set of $\nu\lo k$ at which $X\lo i$ is observed; that is,
$$
J\lo i = \{k: \nu\lo k = t\lo{il~}\text{for some~} l = 1, \ldots, m\lo i \}.
$$
Let $T\lo i = \{\nu\lo k: k \in J\lo i\}$.

Recall that $\ka\lo T$ is a positive definite kernel defined on $T \times T$. 
Let $K\lo T$ be the $m \times m$ Gram matrix whose $(k, l)$th entry is $\ka \lo T (\nu\lo k, \nu\lo l)$. 
Let $\ca H \lo X$ be the RKHS generated by $\{\ka\lo T(\cdot, \nu\lo k): k = 1, \ldots, m\}$.
Then, using the coordinate representation in Section \ref{sec-coor}, 
the inner product between any $f, g \in \ca H \lo X$ can be expressed as
$$
\langle f, g \rangle \lo {\ca H \lo X} = [f]\trans K\lo T [g]. 
$$
As $X\lo i$ is observed at $m\lo i$ distinct time points, we only use functions in $\{\ka(\cdot, \nu\lo k): k \in J \lo i\}$
to represent it; that is,
$$
X\lo i = \sum\lo {k = 1}\hi m [X\lo i]\lo k ~\ka\lo T (\cdot, \nu\lo k) = \sum\lo {k \in J\lo i} [X\lo i] \lo k  ~\ka\lo T (\cdot, \nu\lo k). 
$$
This amounts to setting $[X\lo i]\lo k = 0$ if $k \not\in J\lo i$.
To approximate $X\lo i$ by the observed points $\{X\lo i(\nu\lo l): l \in J\lo i\}$, it suffices to estimate $\{[X\lo i]\lo k: k \in J\lo i\}$.
Let $[X\lo i]\hi 0 = \{([X\lo i]\hi 0)\lo k: k \in J\lo i\}$ denote this the $m\lo i$-dimensional vector, and let $K\lo T \hi {(i, j)}$ denote the $m\lo i \times m\lo j$ sub-matrix with entries $\{\ka\lo T(\nu\lo k, \nu \lo l): k \in J\lo i, l \in J\lo j \}$. Let $X\lo i (T\lo i) = \{X\lo i(\nu\lo k): k \in J\lo i\}$ denote the column vector of dimension $m\lo i$ consisting of the observed points of $X\lo i$. Then 
$X\lo i(T\lo i) = K\lo T \hi {(i, i)} [X\lo i] \hi 0$.  To enhance smoothness when recovering the trajectory of $X\lo i$, we impose the Tikhonov regularization and solve the following equation, which gives
$$
[X\lo i]\hi 0 = \left(K\lo T \hi {(i, i)} + \epsilon \lo T \hi {(X)} I \lo {m\lo i} \right)\inv X\lo i(T\lo i),
$$
where $\epsilon \lo T \hi {(X)}  > 0$ is a tuning parameter. 
Then $X\lo i$ is estimated by $\sum\lo{k \in J \lo i} [X \lo i] \lo k \ka\lo T(\cdot, \nu\lo k)$, where $[X \lo i] \lo k = [X \lo i] \lo k \hi 0$ for $k \in J \lo i$ and $[X \lo i] \lo k = 0$ for $k \notin J \lo i$. For convenience, we still use $X\lo i$ to denote the recovered trajectory of $X\lo i$.
It follows that the inner product between the recovered trajectories is
\begin{align}
	\nonumber
	\langle X\lo i, X\lo j \rangle \lo {\ca H \lo X} \ali  = ([X\lo i]\hi 0 )\trans K\lo T \hi {(i, j)} ([X\lo j]\hi 0 ) \\
	\ali = X\lo i \trans (T\lo i) (K\lo T \hi {(i, i)} + \epsilon \lo T \hi {(X)} I \lo {m\lo i})\inv K\lo T \hi {(i, j)} (K\lo T \hi {(j, j)} + \epsilon \lo T \hi {(X)} I \lo {m\lo j})\inv X\lo j (T\lo j). 
	\label{eq-innX}
\end{align}
The construction of $\ca H \lo Y$ is similar. Let $\mu\lo Y$ denote the Riesz representation of the linear functional $f \mapsto \E\lo n(\langle Y, f\rangle \lo {\ca H \lo Y})$ for $f \in \ca H \lo Y$, where $\E\lo n$ denotes the expectation based on the empirical distribution of $Y$. Obviously, $\mu\lo Y = n \inv \sum\lo {i = 1}\hi n Y\lo i$.

Next, we construct $\frak M \lo X$, which is an RKHS generated by a positive definite kernel on $\ca H \lo X$.
As mentioned in Section \ref{sec-nested}, $\ka(\cdot, \cdot)$ is uniquely determined by the function $\rho: \real\hi 3 \rightarrow \R$ and the inner product in $\ca H \lo X$: for any $u, v \in \ca H \lo X$, 
\begin{align*}
	\ka(u, v)  = \rho \left( \langle u ,  u\rangle \lo {\ca H \lo X}, \langle u,  v\rangle \lo {\ca H \lo X}, \langle v ,  v\rangle \lo {\ca H \lo X} \right),
\end{align*}
where $\langle \cdot, \cdot \rangle \lo {\ca H \lo X}$ is calculated according to \eqref{eq-innX}. In the following implementation, $\ka$ is taken as the Gaussian radial basis function (GRB). The RKHS $\frak M \lo X$ is spanned by $\{\ka (\cdot, X\lo i): i = 1, \ldots, n\}$ with inner product
$$
\langle f, g\rangle \lo {\frak M \lo X} = [f] \trans K\lo X [g], 
$$ for any $f, g \in \frak M \lo X$,
where $K\lo X$ is the $n \times n$ Gram matrix whose $(i, j)$th entry is $\ka  (X\lo i, X\lo j)$.

\subsection{Model fitting}
As indicated in Proposition \ref{prop-condE}, to estimate $\E(Y | X)$, we need to estimate the regression operator $B\lo 0$. Having constructed $\frak M \lo X$, we define $\mu\lo X$ as the Riesz representation of the linear functional $f \mapsto \E\lo n[f(X)]$. 
By Proposition 2 of \cite{li2017}, $\mu\lo X$ is the function $n\inv \sum\lo{i = 1}\hi n\ka(\cdot, X\lo i)$ in $\frak M \lo X$, and $\frak M \lo X \hi 0$ is spanned by $\{\ka(\cdot, X\lo i) - \mu\lo X: i = 1, \ldots, n\}$.

We estimate $\E(Y | X = x)$ by mimicking Equation \eqref{eq-condE} at the sample level. To do so, we next derive the coordinates of relevant operators therein. Here, we omit the associated bases from the notation of coordinate representation as they are obvious from the context. Let $Q = I\lo n - 1\lo n1\lo n \trans/n $, where $1\lo n$ denotes the column vector of length $n$ with each component being 1.
Let $G\lo X = Q K\lo X Q$. Then, by Proposition 3 of \cite{li2017}, 
\begin{align*}
	[\Sigma \lo {XX}] = n \inv G\lo X, \quad \quad [\Sigma \lo {YX}] = n\inv G\lo X,\quad\quad [\Sigma \lo {XX} \mpinv] = n G\lo X\mpinv, 
\end{align*}
with respect to the spanning system $\{\ka(\cdot, X\lo i) - \mu\lo X: i = 1, \ldots, n\}$.
Let $h\lo Y = (Y\lo 1, \ldots, Y\lo n) \trans$. 
Then, given $x \in \ca H \lo X$, 
\begin{align*}
	B\lo 0 \hi * \{\ka (\cdot, x) - \mu\lo X\} & = h\lo Y \trans Q\{[B \lo 0 \hi *] [\ka(\cdot, x) - \mu \lo X]\} \\
	& = h\lo Y \trans Q\left\{[\Sigma \lo {YX}] [\Sigma \lo {XX} \mpinv][\ka(\cdot, x) - \mu \lo X]\right\} \\
	& =  h\lo Y \trans  \{G\lo X G \lo X \mpinv [\ka(\cdot, x) - \mu\lo X]\},
\end{align*}
where the last equality holds because $G \lo X = QK\lo X Q$ and $Q \hi 2 = Q$.
We estimate the Moore-Penrose inverse of $G\lo X$ by the Tikhonov-regularized inverse $(G\lo X + \epsilon\lo X I\lo n)\inv$ to prevent overfitting, where $\epsilon\lo X > 0$ is a tuning constant. It remains to figure out the coordinate of $\ka(\cdot, x) - \mu\lo X$ with respect to the spanning system $\{\ka(\cdot, X\lo i) - \mu \lo X: i = 1, \ldots, n\}$. 
Suppose that $x$ is observed at time points ${t \lo {1}, \ldots, t \lo {m(x)}}$. 
To find the coordinate of $\ka(\cdot, x) - \mu \lo X$,  we first express $x$ as
$$
x = \sum\lo {l = 1}\hi {m(x)} [x]\hi 0 \lo l \ka \lo T (\cdot, t\lo l),
$$
where $[x]\hi 0 = \left(K\lo T (x) + \epsilon\lo T\hi {(x)}I \lo {m(x)} \right) \inv x(T(x))$ with 
$T(x) = \{t\lo 1, \ldots, t\lo{m(x)}\}$, 
$$
x(T(x)) = (x(t\lo 1), \ldots, x(t\lo{m(x)}) )\trans \quad\text{and}\quad  K\lo T (x) = \{\ka\lo T(t\lo i, t\lo j)\}\lo {i, j = 1}\hi {m(x)}. 
$$

Having found the coordinate of $x$, we next identify the coordinate of 
$\ka(\cdot, x) - \mu \lo X$.
Suppose that $[\ka(\cdot, x) - \mu \lo X] = c \lo x$ for some $c \lo x \in \R \hi n$.
Then 
$$
\langle \ka(\cdot, x) - \mu \lo X, \ka(\cdot, X \lo i)  \rangle \lo {\frak M \lo X} = e \lo i \trans K \lo X c \lo x ~ - \frac{1}{n}(e\lo i \trans K \lo X 1 \lo n)(1 \lo n \trans c \lo x) =  e \lo i \trans K Q c\lo x,
$$
where $e\lo i$ denotes the  vector whose $i$th component is 1 and all others are 0. Taking $i = 1, \ldots, n$, we have
$d \lo x = K \lo X Q c \lo x$, where $d \lo x$ is a vector of length $n$ with $i$th component $\ka(X \lo i, x) - \E\lo n \ka(X \lo i, X)$. With the Tikhonov regularization, we obtain the solution $c \lo x = Q(K \lo X + \epsilon \lo X I \lo n)\inv d \lo x$. Lastly, by \eqref{eq-condE}, the predicted value of $y$ is 
\begin{equation}
	\hat{y}(x) = h\lo Y \trans  G \lo X (G \lo X + \epsilon\lo X I\lo n)\inv c \lo x + \frac{1}{n}h\lo Y \trans 1\lo n .
	\label{eq-pred}
\end{equation}

\subsection{Tuning parameter selection} \label{sec:tuning}
This section is concerned with tuning parameters. We have constructed three RKHS's: $\ca H\lo X, \ca H \lo Y$ and $\frak M \lo X$. If we use the GRB as the kernels, then we have tuning parameters: $(\epsilon\lo T\hi{(X)}, \gamma \lo T  \hi {(X)})$,  $(\epsilon\lo T\hi{(Y)}, \gamma \lo T \hi {(Y)})$, and $(\epsilon\lo X, \gamma\lo X)$ for $\ca H \lo X$, $\ca H \lo Y$, and $\frak M \lo X$, respectively. 

Since constructions of $\ca H \lo X$ and $\ca H \lo Y$ are essentially the same, we only illustrate the choice of
$(\epsilon\lo T\hi{(X)}, \gamma\lo T\hi{(X)})$. From the construction for $\ca H \lo X$ in Section \ref{sec-const}, we see that the predicted value of $X\lo i$ at any $t \in T$ is given by
$$
\hat{X}\lo i(t; \gamma\lo T \hi{(X)}, \epsilon\lo T\hi{(X)}) = [X\lo i(T\lo i)]\trans \left\{K\lo T \hi {(i, i)}(\gamma \lo T\hi{(X)}) + \epsilon \lo T \hi {(X)} I \lo {m\lo i} \right\}\inv \ka\lo T(t, T\lo i; \gamma \lo T\hi{(X)}), 
$$
where $\ka\lo T(t, T\lo i; \gamma\lo T\hi{(X)}) = (\ka\lo T(t, t\lo{i1}), \ldots, \ka\lo T(t, t\lo {im\lo i}) )\trans$. Since the function $\hat{X}\lo i$ can be viewed as a linear smoother from the perspective of nonparametric smoothing, we suggest using the generalized cross validation (GCV) to choose the optimal $(\epsilon\lo T\hi{(X)}, \gamma\lo T\hi{(X)})$. Specifically, let
$$
\text{GCV}(\epsilon\lo T\hi{(X)}, \gamma\lo T\hi{(X)}) := \sum\lo {i = 1}\hi n \frac{m\lo i\inv \sum\lo{j = 1}\hi{m\lo i} [X\lo i(t\lo {ij}) - \hat{X}\lo i(t\lo {ij}; \gamma\lo T\hi{(X)}, \epsilon\lo T\hi{(X)}) ]\hi 2} {\{1 - \tr[S\lo i(\epsilon\lo T\hi{(X)},\gamma\lo T\hi{(X)})]/m\lo i\}\hi 2},
$$
where $S\lo i(\epsilon\lo T\hi{(X)},\gamma\lo T\hi{(X)}) = K\lo T \hi {(i, i)} (\gamma\lo T\hi{(X)}) \left\{K\lo T \hi {(i, i)}(\gamma\lo T\hi{(X)}) + \epsilon \lo T \hi {(X)} I \lo {m\lo i} \right\} \inv$ is the 
smoother matrix for $X\lo i, i = 1, \ldots, n$. The optimal $(\epsilon\lo T\hi{(X)}, \gamma\lo T\hi{(X)})$ is chosen by minimizing the GCV score over a grid of $(\epsilon\lo T\hi{(X)}, \gamma\lo T \hi {(X)})$.

Similarly, we also choose the tuning parameters $(\epsilon\lo X, \gamma\lo X)$ by GCV. By \eqref{eq-pred}, the fitted value of $Y\lo i$ at $X\lo i$ is $$
\hat{Y}(X \lo i) = [Q\lo i\hi\trans G\lo X (G\lo X + \epsilon\lo X I\lo n) \inv  + 1\lo n\hi\trans/n] h\lo Y,
$$
where $Q\lo i = Qe \lo i$ is the $i$th column of the projection matrix $Q$. Therefore, the GCV score in this case is defined as
$$
\text{GCV}(\epsilon\lo X, \gamma\lo X) = \frac{1}{n}\sum\lo {i=1}\hi n \frac{\|Y\lo i - \hat{Y}\lo i\|\lo{\ca H \lo Y}\hi 2}{\{1 - \tr[QG\lo X(G\lo X + \epsilon\lo X I\lo n) \inv + 1\lo n1\lo n\hi\trans/n]/n\}\hi 2}.
$$
The optimal $(\epsilon\lo X, \gamma\lo X)$ is chosen by minimizing GCV over a grid of $(\epsilon\lo X, \gamma\lo X)$.

\section{Convergence rates}\label{section:convergence rate}

In this section we develop the convergence rates of our nonparametric regression. In particular, we are interested in  the the following two rates:
\begin{enumerate}
	\item the convergence rate of the estimated regression operator $\hat B$;
	\item the convergence rate of the regression estimate $\widehat \E (Y | x \lo 0 )$ at a new predictor and at  any time point $t$.
\end{enumerate}
We will also derive the optimal tuning parameter $\epsilon \lo n = \epsilon \lo X$ that makes these rates the fastest.

\subsection{Some preliminary  lemmas}

Let $\hat \Sigma \lo {XX}$ and $\hat \Sigma \lo {XY}$ be the estimates of $\Sigma \lo {XX}$ and $\Sigma \lo {XY}$ as defined in Section \ref{sec-estimation}. We first introduce  some notations about linear operators. Let $\ca G \lo 1$ and $\ca G \lo 2$ be two generic separable Hilbert spaces and $A: \ca G \lo 1 \to \ca G \lo 2$ a linear operator. Then $A$ is a Hilbert-Schmidt operator if $\sum \lo {i \in \nat} \, \| A  e \lo i \| \lo {\ca G \lo 2} \hi 2 < \infty$, where $\nat = \{1, 2, 3, \ldots\}$, and $\{ e \lo i : i \in \nat \}$ is any orthonormal basis (ONB) of $\ca G \lo 1$. The square root of this finite number  is the Hilbert-Schmidt norm,  and is denoted by $\| A \| \loo {HS}$. We will use $\| \cdot \| \loo {OP}$ to denote the operator norm. Given two arbitrary positive sequences $\{a \lo n: n \in \nat\}$ and $\{ b \lo n: n \in \nat \}$, we write $a \lo n \prec b \lo n$ if $a \lo n / b \lo n \to 0$, write $ a\lo n \succ b\lo n$ if $b \lo n \prec a \lo n$, write $a \lo n \preceq b \lo n$ if $a \lo n / b\lo n$ is a bounded sequence and write $a \lo n \asymp b \lo n$ if $a \lo n \preceq b \lo n$ and $b \lo n \preceq a \lo n$. For two real numbers $a$ and $b$, we use $a \wedge b$ to represent the minimum of $a$ and $b$. We make the following assumption. 

\setcounter{assumption}{3}
\begin{assumption}\label{assumption:finte moments and beta}  \ \
	\begin{enumerate}
		\item[(i)]  $\E[ \ka (X, X) ] < \infty$, $\E ( \| Y \| \lo {\ca H \lo Y} \hi 2) < \infty$;
		\item[(ii)] there is a $\beta > 0$ such that $\Sigma \lo {XY} = \Sigma \lo {XX} \hi {1+\beta} S \lo {XY}$ for some bounded linear operator $S \lo {XY}: \ca H \lo Y \to \frak M \lo X$.
	\end{enumerate}
\end{assumption}

It can be shown that, under the assumption $\E [\ka (X, X)] < \infty$, $\Sigma \lo {XX}$ is a trace-class operator.
As argued in \cite{li2017} and \cite{li2018}, Assumption \ref{assumption:finte moments and beta}(ii)  represents a degree of smoothness in the relation between $X$ and $Y$. It requires the output functions of $B \lo 0$ to be sufficiently concentrated on the low-frequency components of $\Sigma \lo {XX}$.
Indeed, if $\{(\lambda \lo j, \varphi \lo j ): j \in \nat \}$ is the eigenvalue-eigenfunction sequence of $\Sigma \lo {XX}$ with $\lambda \lo 1 \ge \lambda \lo 2 \ge \cdots$, then $\Sigma \lo {XY} = \Sigma \lo {XX} \hi {1+\beta} S \lo {XY}$ implies that, for any $g \in \ca H \lo Y$,
\begin{align}\label{eq:equivalent condtion for beta}
	\tsum \lo {j \in \nat} \, \lambda \lo j \hi {-2\beta}\, { \langle B \lo 0 g, \varphi \lo j \rangle \lo {\frak M \lo X} \hi 2}< \infty.
\end{align}
The following lemma gives the convergence rates of $\hat \Sigma \lo {XX}$ and $\hat \Sigma \lo {XY}$, whose proof is similar to that of Lemma 5 of \cite{fukumizu2007} and is omitted.

\begin{lemma}\label{lemma:fuku et al} Under Assumption \ref{assumption:finte moments and beta}(i),   $\Sigma \lo {XX}$ and $\Sigma \lo {XY}$ are Hilbert-Schmidt operators and
	\begin{align*}
		\| \hat \Sigma \lo {XX} - \Sigma \lo {XX} \| \loo {HS} = O \lo P ( n \hi {-1/2}), \quad \| \hat \Sigma \lo {XY} - \Sigma \lo {XY} \| \loo {HS} = O \lo P ( n \hi {-1/2}).
	\end{align*}
\end{lemma}

Let $\hat B= ( \hat \Sigma \lo {XX} + \epsilon \lo n I ) \inv \hat \Sigma \lo {XY}$ denote the sample estimator of $B \lo 0$ in Section \ref{sec-estimation}, where we have used $\epsilon \lo n$ to replace $\epsilon \lo X$ to highlight the dependence on the sample size $n$. Under Assumption \ref{assumption:finte moments and beta}, the best convergence rate of $\hat B$ to $B \lo 0$ developed by  \cite{li2017} is $n \hi {-\beta/[2 (\beta +1)]}$. If $\beta=1$, this rate reaches its fastest possible level $ n \hi {-1/4}$. In the next subsection we will show that, in our regression setting and with an additional assumption on $\Sigma \lo {XX}$, the convergence rate of $\hat B$ can approach  $n \hi {-1/3}$.

Let $U = Y - \E(Y|X)$ be the population-level residual, which is a random element in $\ca H \lo Y$. Let $\Sigma \lo {XU} = \E[ ( \ka (\cdot, X) - \mu \lo X)\otimes U ]$. Let $\hat \mu \lo U$ and $\hat \Sigma \lo {XU}$ be the sample estimates of $\mu \lo U$ and $\Sigma \lo {XU}$ defined by
\begin{align*}
	\hat \mu \lo U = \E \lo n (U), \quad \hat \Sigma \lo {XU} = \E \lo n [ (\ka (\cdot, X) - \hat{\mu} \lo X) \otimes (U - \hat \mu \lo U)].
\end{align*}

\begin{lemma}\label{lemma:reg and res} Under Assumption \ref{assumption:finte moments and beta}(i),
	\begin{enumerate}
		\item[{\em (1.)}] $\Sigma \lo {XU} = 0$;
		\item[{\em (2.)}] $\hat \Sigma \lo {XY} = \hat \Sigma \lo {XU} + \hat \Sigma \lo {XX} B \lo 0$.
	\end{enumerate}
\end{lemma}

Let
$\tilde \Sigma \lo {XU}= E \lo n [ ( \ka (\cdot, X) - \mu \lo X) \otimes U]$, which  is an intermediate operator between $\hat \Sigma \lo {XU}$ and $\Sigma \lo {XU}$.
\begin{lemma}\label{lemma:hat vs tilde} Under Assumption \ref{assumption:finte moments and beta}(i), we have
	\begin{align*}
		\| \hat \Sigma \lo {XU} - \tilde \Sigma \lo {XU} \| \loo {HS} = O \lo P(n \inv).
	\end{align*}
\end{lemma}

Since $\Sigma \lo {XX}$ is a trace-class operator under Assumption \ref{assumption:finte moments and beta}(i), we have $\tsum \lo {j \in \nat} \lambda \lo j < \infty$. The next  assumption strengthens this condition. It also strengthens the condition $\Sigma \lo {XU} = 0$. 

\setcounter{assumption}{4}
\begin{assumption}\label{assumption:KL for kappa}  \ \
	\begin{enumerate}
		\item[(i)] $X \indep U$;
		\item[(ii)] $\lambda \lo  j \asymp j \hi {-\alpha}$ for some $\alpha > 1$.
	\end{enumerate}
\end{assumption}
\noindent Part (i) of this assumption would be satisfied if our function-on-function regression model is
\begin{align}\label{eq:RKHS regression}
	Y = f(X) + U,
\end{align}
where $Y$ and $U$ are random elements in $\ca H \lo Y$, $X$ is a random element in $\ca H \lo X$, and $f$ is a (nonlinear) mapping from $\ca H \lo X$ to $\ca H \lo Y$.
Part (ii) of this assumption is about the niceness of the random function $X$: its variation is concentrated on the low-frequency domain of the spectrum of the covariance operator $\Sigma \lo {XX}$. The next lemma reveals how Assumption \ref{assumption:KL for kappa}(ii) interacts with Tychonoff regularization.

\begin{lemma}\label{lemma:calculus} Under Assumption \ref{assumption:KL for kappa}({ii}), if $\epsilon \lo n \prec 1$, then 
	$$
	\tsum \lo {j \in \nat}  \lambda \lo j  ( \lambda \lo j + \epsilon \lo n ) \hi {-2} = O   (\epsilon \lo n \hi {-(\alpha+1)/ \alpha }).
	$$
\end{lemma}

\subsection{Convergence rate for estimated regression operator}

For convenience, we abbreviate $(\hat \Sigma \lo {XX} + \epsilon \lo n I) \inv$, $( \Sigma \lo {XX} + \epsilon \lo n I) \inv$, and $\Sigma \lo {XX} \hi\dagger$ by $\hat V$, $V \lo n$ and $V$, respectively. The following Fourier expansion of $\ka (\cdot, X) - \mu \lo X$ with respect to the eigenfunction orthonormal basis (ONB) $\{\varphi \lo j: j \in \nat \}$ will be useful:
\begin{align}\label{eq:kappa expansion}
	\ka (\cdot, X) - \mu \lo X = \tsum \lo {j \in \nat} \, \langle \ka (\cdot, X) - \mu \lo X, \varphi \lo j \rangle \lo {\frak M \lo X} \,   \varphi \lo j \equiv \tsum \lo {j \in \nat} \, \zeta \lo j  \varphi \lo j,
\end{align}
where $\zeta \lo 1, \zeta \lo 2, \ldots$ are uncorrelated variables with $\E(\zeta \lo j) = 0$ and $\var ( \zeta \lo j) = \lambda \lo j$.

\begin{theorem}\label{theorem:convergence rate} Suppose Assumptions 1 through 3 hold; Assumption 4 holds for some $\beta > 0$; Assumption 5 holds for some $\alpha > 1$; $\epsilon \lo n \prec 1$.
	Then
	\begin{enumerate}
		\item [{\em (1.)}]
		\begin{align}\label{eq:goal 1}
			\| \hat B - B \lo 0 \| \loo {OP} = O \lo P ( n \hi {-1/2} \epsilon \lo n \hi {(\beta \wedge 1) - 1} + \epsilon \lo n \hi {\beta\wedge 1}  +  n\hi {-1} \epsilon \lo n \hi {- (3 \alpha+ 1)/ (2 \alpha)} + n \hi {-1/2} \epsilon \lo  n \hi {-(\alpha + 1)/(2 \alpha)} ).
		\end{align}
		\item [{\em (2.)}]
		If
		$
		\epsilon \lo n \succ \max ( n \hi {- 1/ [2\{1-(\beta \wedge 1)\}]}, n \hi {-2 \alpha / ( 3 \alpha + 1)} ),
		$
		then the right-hand side of (\ref{eq:goal 1}) tends to 0.
	\end{enumerate}
\end{theorem}

%
%
%

%

\subsection{Optimal turning and convergence}

Next, we derive the optimal convergence rate of (\ref{eq:goal 1}) where $\epsilon \lo n$ is of the form $\epsilon \lo n \asymp n \hi {-\delta}$ for some $\delta > 0$. With $\epsilon \lo n$ in this form, the four terms in (\ref{eq:goal 1}) reduce to
\begin{align*}
	n \hi {-1/2+\delta \{1-(\beta \wedge 1)\}}, \quad   n \hi {- \delta (\beta \wedge 1)}, \quad   n\hi {-1+\delta (3 \alpha+ 1)/ (2 \alpha)}, \quad  n \hi {-1/2+  \delta (\alpha + 1)/(2 \alpha)}.
\end{align*}
Let
$\ell  \lo 1, \ldots, \ell \lo 4$ be the  linear functions of $\delta$ in the exponents; that is,
\begin{align*}
	\ell \lo 1 (\delta)= {-1/2+\delta \{1-(\beta \wedge 1)\}}, \quad   \ell \lo 2 (\delta)= {- \delta (\beta \wedge 1)}, \quad  \ell \lo 3(\delta)= {-1+\delta (3 \alpha+ 1)/ (2 \alpha)}, \\
	\ell \lo 4(\delta)= {-1/2+  \delta (\alpha + 1)/(2 \alpha)}. \hspace{1.4in}
\end{align*}
Let $m(\delta) = \max \{ \ell \lo 1(\delta), \ldots, \ell \lo 4(\delta) \}$. Then the rate in (\ref{eq:goal 1}) can be rewritten as $n \hi {m (\delta)}$.
Letting $\delta \loo {opt}$ be the $\delta$ that minimizes  $m(\delta)$,  the optimal tuning parameter is $\epsilon \lo n =   n \hi {- \delta \loo {opt}}$, and the corresponding convergence rate  is $n \hi {m(\delta \loo {opt})} \equiv \rho \loo{opt}$.

\begin{theorem}\label{theorem:optimal rate} Suppose the conditions in Theorem \ref{theorem:convergence rate} hold for some $\alpha > 1$, $\beta > 0$.
	\begin{enumerate}
		\item [({\em 1})] if $\beta > (\alpha - 1)/(2 \alpha)$, then
		$
		\delta \loo {opt} = \frac{\alpha}{2 \alpha (\beta \wedge 1) + \alpha + 1}$, $\rho \loo {opt} = n \hi {-\frac{\alpha (\beta \wedge 1)}{2 \alpha (\beta \wedge 1) + \alpha + 1}}.
		$
		\item [({\em 2})] if $\beta \le (\alpha - 1)/(2 \alpha)$, then
		$
		\delta \loo {opt} = \frac{1}{2}$, $\rho \loo {opt} = n \hi {- \frac{\beta}{2}}.
		$
	\end{enumerate}
\end{theorem}

The best rate for the regression operator reported in \cite{li2017} is
\begin{align*}
	\rho \loo {LS} = n \hi {-(\beta \wedge 1) / [2 \{1 + (\beta \wedge 1)\}]}.
\end{align*}
It is easy to check that $\rho \loo {opt}$ converges to 0 faster than $\rho \loo {LS}$ in both scenarios of $\beta$; that is,
\begin{align*}
	n \hi {-{(\alpha \beta \wedge \alpha)}/{(2 \alpha (\beta \wedge 1) + \alpha + 1)}} \prec n \hi {-(\beta \wedge 1) / [2 \{1 + (\beta \wedge 1)\}]}, \quad n \hi {- {\beta}/{2}}\prec n \hi {-\beta / [2 ( 1 + \beta)]}
\end{align*}
for all $\beta > 0$ and $\alpha > 1$. The reason for this improvement is twofold: first, we are dealing with the  more specific regression problem (\ref{eq:RKHS regression}), whereas \cite{li2017} dealt with a general problem where the regression operator corresponds directly to a conditional distribution, without any regression structure; second, we have made Assumption \ref{assumption:KL for kappa}(ii), which was not made in \cite{li2017}.  Note that, when $\beta = 1$, Li and Song's rate is $n \hi {-1/4}$, whereas our current rate is always faster than $n \hi {-1/4}$ regardless of the value of $\alpha$, and approaches $n \hi {-1/3}$ when $\alpha$ is large.

\subsection{Convergence rate for regression estimate}
In this section we develop the convergence rate of our nonparametric regression estimate
$\widehat \E (Y | x \lo 0 )$ to the true mean response $\E ( Y | x \lo 0)$ at any given time point $t$.
We will use  $\E ( Y | x \lo 0)(t)$ to denote  the function $\E(Y|X= x \lo 0)$, which is a member of $\ca H \lo Y$, evaluated at time $t$; the same  applies to $\widehat \E (Y | x \lo 0 )(t)$.
Assuming  $\ca H \lo Y$ is an RKHS with kernel $\ka \lo T$, the conditional mean $\E ( Y | x \lo 0)(t)$ can be written as $\langle \ka \lo T (\cdot, t), \E ( Y | x \lo 0) \rangle \lo {\ca H \lo Y}$. Since
$
\E(Y| x\lo 0) = B \lo 0 \hi * ( \ka (\cdot, x \lo 0) - \mu \lo X) + \mu \lo Y,
$
we have
\begin{align}\label{eq:true regression function}
	\E ( Y | x \lo 0)(t) = \ali  \langle B \lo 0 \, \ka \lo T (\cdot, t),    \ka (\cdot, x \lo 0) - \mu \lo X  \rangle \lo {\frak M \lo X}  +   \mu \lo Y (t).
\end{align}
The estimate of the above is
\begin{align}\label{eq:estimated regression function}
	\widehat \E ( Y | x \lo 0)(t)
	= \langle \hat B \, \ka \lo T (\cdot, t),  \ka (\cdot, x \lo 0) - \hat \mu \lo X \rangle \lo {\frak M \lo X} +  \hat \mu \lo Y (t).
\end{align}
The next corollary shows that  $\widehat \E (Y | x \lo 0 )(t)$ has the same convergence rate as $\hat B$.

\begin{corollary}\label{corollary:regression estimate} Suppose
	\begin{enumerate}
		\item [({\em 1})]the conditions in Theorem \ref{theorem:convergence rate} hold for some $\alpha > 1$, $\beta > 0$,
		\item [({\em 2})]
		$\max ( n \hi {- 1/ [2\{1- (\beta \wedge 1)\}]}, n \hi {-2 \alpha / ( 3 \alpha + 1)} )\prec \epsilon \lo n \prec 1 $;
		\item [({\em 3})]
		$\ca H \lo Y$ is an RKHS generated by a kernel $\ka \lo T$.
	\end{enumerate}
	Then $\widehat \E (Y | x \lo 0) (t)$ is consistent with convergence rate
	\begin{align}\label{eq:hat E rate}
		\widehat \E (Y | x \lo 0) (t) - \E (Y | x \lo 0) (t)  =
		O \lo P ( n \hi {-1/2} \epsilon \lo n \hi {(\beta \wedge 1) - 1} + \epsilon \lo n \hi {\beta\wedge 1}  +  n\hi {-1} \epsilon \lo n \hi {- (3 \alpha+ 1)/ (2 \alpha)} + n \hi {-1/2} \epsilon \lo  n \hi {-(\alpha + 1)/(2 \alpha)} ).
	\end{align}
	Furthermore, the conclusions of Theorem \ref{theorem:optimal rate} also hold.
\end{corollary}

\section{Central limit theorem} \label{sec-CLT}

\def\cid{\stackrel{\ca D}\longrightarrow}

\subsection{Pointwise central limit theorem} 
In this section we develop the central limit theorem of the regression estimate $\widehat \E (Y | x \lo 0) (t)$, which is useful for constructing the confidence interval for the mean response $\E(Y|x \lo 0) (t)$. We will only consider the case $\beta > (\alpha-1)/(2 \alpha)$ and $\delta > \alpha / ( 2 \alpha \beta + \alpha + 1)$, which means the relation between $Y$ and $X$ is relatively smooth and $\epsilon \lo n$ is chosen so that the bias term is of a smaller order than the dominating term. More specifically, recall that
\begin{align*}
	\hat B - B \lo 0
	=  \ali  \hat B \loo {res} +   (\hat B \loo {reg} - B \lo n)  +   (B \lo n - B \lo 0 ), \quad \mbox{where} \\
	\hat B \loo {res}=  \ali (\hat V \hat \Sigma \lo {XU} - \hat V \tilde \Sigma \lo {XU}  )  + (\hat V \tilde \Sigma \lo {XU} - V \lo n \tilde  \Sigma \lo {XU} ) +   V \lo n \tilde  \Sigma \lo {XU}.
\end{align*}
Let $B\lo {n,4}$ and $B \lo {n,5}$ be the last two terms of the first equation, and $B \lo {n,1}, B \lo {n,2}, B \lo {n,3}$ be the three terms of $\hat B \loo {res}$ in the second equation. Let
\begin{align*}
	A \lo {n,r} = \langle B \lo {n, r} \ka \lo T (\cdot, t), \ka (\cdot, x \lo 0) - \mu\lo X \rangle \lo {\frak M \lo X}, \quad r = 1, \ldots, 5.
\end{align*}
Note that $A \lo {n,5}$ is a nonrandom number. By Theorem \ref{theorem:convergence rate},  when $\beta > (\alpha-1)/(2 \alpha)$ and $\alpha / ( 2 \alpha \beta + \alpha + 1) < \delta < 1/2$, $B \lo {n,3}$ is the dominating term among all the other terms. Hence it is reasonable to expect that
$A \lo {n,3}$ is also the dominating term. Our central limit theorem is based on this assumption. 

\begin{assumption}\label{assumption:dominating}  $A \lo {n,1}, \ldots, A \lo {n,4}$ have finite variances $\sigma \lo {n,1} \hi 2, \ldots, \sigma \lo {n,4} \hi 2$ and
	\begin{align*}
		\sigma \lo {n,3}  \succ \max ( \sigma \lo {n, 1} , \sigma \lo {n,2} , \sigma \lo {n,4}, |A \lo {n, 5}|).
	\end{align*}
\end{assumption}

\medskip

\begin{theorem}\label{theorem:clt} Suppose the conditions in Theorem \ref{theorem:convergence rate} are satisfied for some $\alpha > 1$ and $\beta > (\alpha-1)/(2 \alpha)$, and Assumption \ref{assumption:dominating} is satisfied. Furthermore, suppose that the kernel $\ka$ is bounded. Then the following statements hold true:
	\begin{enumerate}
		\item [({\em 1})] $\sigma \lo {n,3} \hi 2 =n \inv \E [ U \hi 2 (t)] \, \tsum \lo {j \in \nat} (\lambda \lo j + \epsilon \lo n)\hi {-2} \lambda \lo j
		[\varphi \lo j (x \lo 0)]\hi 2$;
		\item [({\em 2})] if $\sigma \lo {n,3} \succ  \epsilon \lo n \hi {\beta \wedge 1}$ and $\epsilon \lo n \succ n \hi {-1/2}$, then for any $x \lo 0 \in \ca H \lo X$ and $t \in T$,
		\begin{align*}
			\sigma \lo {n, 3} \hi {-1} [ \widehat \E (Y|x \lo 0)(t) - \E (Y|x \lo 0) (t) ] \stackrel{\ca D}\longrightarrow N(0,1).
		\end{align*}
	\end{enumerate}
\end{theorem}

To use this theorem to construct confidence intervals, we need to have an estimate of $\sigma \lo {n,3} \hi 2$.
As will be discussed later, we can substitute the estimates of $\lambda \lo j$, $\varphi \lo j$ and $\E[ U(t) \hi 2]$ to estimate $\sigma \lo {n,3}\hi 2$ for constructing the confidence interval.

\subsection{Uniform central limit theorem} \label{sec:unifclt}

Following the idea of \cite{cardot2007},  we now study the weak convergence of the regression estimate as a random function in the Hilbert space $\ca H \lo Y$.
With a slight abuse of notation, we denote the Riesz representation of $T\lo x$ defined in Section \ref{sec-model} by $M(x)$ given $X = x$ in $\ca H \lo X$, which is actually $\E(Y | X = x) \in \ca H \lo Y$. Let $\widehat{M}(x)$ denote the predicted value in $\ca H \lo Y$ for a new value $x$ obtained by means of the estimation method introduced in Section \ref{sec-estimation}. We are interested in the following problem. Given a new random element $X\lo {n + 1} \in \ca H \lo X$ that is a copy of $X$ and independent of $X\lo 1, \ldots, X\lo n$, we aim to investigate the weak convergence of $ a\lo n [\widehat{M}(X\lo {n + 1}) - \E(Y\lo {n + 1} | X\lo {n + 1})]$
in $\ca H \lo Y$ for some normalizing constant $a\lo n$. The following lemma illustrates the stochastic order of the crucial term in establishing weak convergence of $\widehat{M}(X\lo {n + 1})$.

\begin{lemma} \label{le-main}
	Suppose the conditions in Theorem \ref{theorem:convergence rate} are satisfied for some $\alpha > 1$, $\beta \geq 1$, and $n\hi {-1/2} \prec \epsilon\lo n \prec 1$. We further assume that $S\lo{XY}$ in Assumption \ref{assumption:finte moments and beta}(ii) is a Hilbert-Schmdit operator. 
	Let $W\lo n =  \sum\lo {i = 1}\hi n Z\lo {i, n}$, where $Z\lo {i,n} = \frac{1}{n} U\lo i  \langle \hat{V} G\lo i, G\lo {n+1} \rangle \lo {\frak M \lo X}$ and $G\lo i = \ka(\cdot, X\lo {i}) - {\mu}\lo X$.
	Then the following statements hold true:
	\begin{enumerate}
		\item[{\em 1.}] $\widehat{M}(X\lo {n + 1}) - \E(Y\lo {n + 1} | X\lo {n + 1}) = W\lo n +  O\lo P\left(n\hi {-1/2}\epsilon\lo n\hi{(\alpha-1)/(2\alpha)} + \epsilon\lo n + n\hi{-1} \epsilon\lo n\hi {-(\alpha+1)/(2\alpha)}  \right)$;
		\item[{\em 2.}] $\E(\|W\lo n\|\lo {\ca H \lo Y}\hi 2 ) = O(n\inv\epsilon\lo n\inv)$.
		
	\end{enumerate}
	
\end{lemma}

\begin{remark} \label{rem:beta}
	We impose the condition $\beta \geq 1$ to facilitate the analysis of the stochastic order of $\hat{\Sigma}\lo {XX}\hi {\beta} - \Sigma\lo {XX}\hi {\beta}$. Without this assumption,
	even though we can still prove that it is $o\lo p(1)$, determining its convergence rate is quite complicated. 
\end{remark}

By \eqref{eq-Mhatexpansion} (in the appendix) in the proof of Lemma \ref{le-main} part {\em 1}, after ignoring $O\lo P(n\hi{-1/2})$ terms, we have
\begin{align*}
	\ali\quad\quad\widehat{M}(X\lo {n + 1}) - \E(Y\lo {n+1} | X\lo {n + 1}) \\
	\ali = (\hat{\Sigma}\lo {YX}\hat{V} - \Sigma\lo{YX}\Sigma\lo {XX}\mpinv)\{\ka(\cdot, X\lo {n+1}) - {\mu}\lo X\} \\
	\ali = B\lo 0\hi*(\hat{\Sigma}\lo{XX}\hat{V}-\Sigma\lo{XX}V\lo n)G\lo{n+1} + 
	B\lo 0\hi* ({\Sigma}\lo{XX}V\lo n - I) G\lo {n+1}   +  \frac{1}{n} \sum\lo {i = 1}\hi n U\lo i \langle \hat{V} \tilde{G}\lo i, G\lo {n+1} \rangle \lo {\frak M \lo X}   \\
	& = F\lo{1n} + F\lo{2n} + W\lo n + R\lo n,
\end{align*}
where $R\lo n = \bar{U} \langle \hat{V} (\mu\lo X - \hat{\mu}\lo X), G\lo {n+1} \rangle \lo {\frak M \lo X} = O\lo P(n\hi {-1} \epsilon\lo n\hi {-(\alpha+1)/(2\alpha)})$. 
It is straightforward to check that $n\hi{-1/2}\epsilon\lo n \hi {-1/2} \succ \max(n\hi {-1/2}\epsilon\lo n\hi {(\alpha-1)/(2\alpha)}, \epsilon\lo n, n\inv \epsilon\lo n\hi {-(\alpha+1)/(2\alpha)} )$ if $n\hi {-1/2} \prec \epsilon\lo n \prec n \hi {-1/3}$ for $\alpha > 1$. Thus $W\lo n$ is the dominating term among all the other terms when $n\hi {-1/2} \prec \epsilon\lo n \prec n \hi {-1/3}$. 
Let $s\lo n \hi 2 = n\inv  \E[\langle \hat{V} G\lo i, G\lo{n+1} \rangle \lo {\frak M \lo X}\hi 2]$. The  weak convergence of $\widehat{M}(X\lo {n + 1})$ in $\ca H \lo Y$  is based on the assumption that $W\lo n$ is the dominating term.

\begin{assumption} \label{assumption: uniform dominating}
	\begin{align*}
		s\lo n \hi 2 \succ \max(\var(F \lo{1n}), \var(F \lo{2n}), \var(R\lo n)). 
	\end{align*}
\end{assumption}

\begin{theorem} \label{th-strongCLT}
	Suppose assumptions in Lemma \ref{le-main} are met and Assumption \ref{assumption: uniform dominating} is satisfied. 
	Then
	\begin{equation}
		s\lo n\hi {-1}[\widehat{M}(X\lo {n + 1}) - \E(Y\lo {n+1} | X\lo {n + 1}) ]
		\convd  \ca N,
		\label{eq-strongCLT}
	\end{equation}
	where $\ca N$ is a centered Gaussian element taking values in $\ca H \lo Y$ with covariance operator $\Sigma\lo {UU}$. 
\end{theorem}

\begin{remark} \label{re:SCB}
	By \eqref{eq-strongCLT} and the continuous mapping theorem, we have
	$$
	\sup\lo {t \in T} \biggl|s\lo n\hi {-1}\left[\{\widehat{M}(X\lo {n + 1})\}(t) - \E(Y\lo {n+1}  | X\lo {n + 1}) (t)\right]\biggl| \stackrel{\ca D}\longrightarrow \sup\lo {t \in T} |\ca N(t)|$$
	as $n \rightarrow \infty$. Therefore, if we are able to find $C(\alpha)$ that satisfies $\Pr (\sup\lo {t \in T} |\ca N(t)| \leq C(\alpha)) = 1 - \alpha$, then a $(1 - \alpha)$ simultaneous confidence band for $\E(Y\lo {n+1} | X\lo {n+1})$ can be constructed as
	\begin{equation} \label{eq-simulCI}
		\left(\{\widehat{M}(X\lo {n + 1})\}(t) - s\lo n C(\alpha), \{\widehat{M}(X\lo {n + 1})\}(t) + s\lo n C(\alpha)\right).
	\end{equation}
	The determination of $C(\alpha)$ is illustrated in one of our simulation studies near the end of Section \ref{sec:densesimul}.
\end{remark}

%

\section{Simulation studies} \label{sec-simu}
In this section, we investigate the performance of the proposed methodology in prediction under different simulation scenarios. For this purpose, we compare our nonlinear function-on-function regression (to be abbreviated by NLFFR) method with several alternative methods: optimal penalized linear function-on-function regression (to be abbreviated by PLFFR) proposed by \cite{sun2018optimal} and linear function-on-function regression estimated via functional
principal component analysis (to be abbreviated by FPCA) proposed by \cite{yao2005b} and \cite{crambes2013}.
In addition, we evaluate the finite-sample performances of both the pointwise confidence interval and the simultaneous confidence band developed in Section \ref{sec-CLT}.

\subsection{Simulation of functional covariate and functional response}
We adopt a similar strategy in \cite{li2017} to generate functional covariates. Specifically, we construct $\ca H \lo X$ as the RKHS induced by two kernels: the Gaussian radial basis function (GRB) and the Brownian motion covariance function (BMC). When the GRB kernel is employed, the functional covariate $X$ is generated by
$X(\cdot) = \sum\lo {k = 1}\hi 5 a\lo k \ka\lo T(\cdot, t\lo k)$, where $a\lo 1, \ldots, a\lo 5$ are independently sampled from $N(0, 1)$, $t\lo 1, \ldots, t\lo 5$ are independently sampled from $U[0, 1]$ and $\gamma\lo T = 7$. When the BMC kernel is employed, $X$ is generated as 
$$
X(t) = \sum\lo {j = 1}\hi {100} \sqrt{2} [(j - 1/2)\pi]\hi {-1} a\lo j \sin [(j - 1/2)\pi t], 
$$
where $a\lo j$'s are independently sampled from $N(0, 1)$. For each kernel, we consider both dense and sparse designs for the observed time points of $X$. In the dense design, we choose 50 equally spaced points in $[0, 1]$ as the observed time points of $X$ for each subject, while in the sparse design, we randomly select 10 points from the aforementioned 50 equally spaced points for each subject. The left panel of Figure \ref{fig:design} depicts 10 sample paths of $X$ generated by these two kernels in the dense design. 

Two models are then used to generate the functional response $Y$:
\begin{align*}
	\mbox{Model}~~1:\quad Y(t) & = \left(\frac{1}{1 + e\hi {\langle X, b\lo 1 \rangle\lo {\ca H \lo X}}} + \langle X, b\lo 2  \rangle\lo {\ca H \lo X} \hi 2\right) \rho(t) + \sigma \epsilon(t), \\
	\mbox{Model}~~2: \quad Y(t) & = \left\{\cos\left(\langle X, b\lo 3\rangle\lo{\ca H \lo X}\right)\right\}\rho(t)+ \sigma \epsilon(t). 
\end{align*}
In both models, when the GRB kernel is used, $b\lo j (\cdot) = \ka\lo T(\cdot, t\lo j)$ with $t\lo 1 =0.6,~t\lo 2 = 0.9$ and $t\lo 3 = 0.1$; when the BMC kernel is used, for $j = 1, 3$,
$b\lo j(t) = \nu\lo j(t) : = \sqrt{2} a\lo j \sin [(j - 1/2)\pi t]$, which is actually the $j$th eigenfunction of the covariance operator of the standard Brownian motion, and $b \lo 2 (t) = 0$. Regardless of the choice of kernel, $\rho(t) = \sum\lo {j = 1}\hi 5 \nu\lo j(t)$ and $\epsilon(t)$ is generated from the standard Brownian motion. The choices of $\rho$ and $\epsilon$ ensure that the true conditional mean $E(Y | X)$ resides in the RKHS generated by the BMC kernel. 
The right panel of Figure \ref{fig:design} shows the shape of $\rho(t)$, which indicates that the (true) conditional mean has a relative large fluctuation around 0.18. 
We consider two different values of $\sigma$: 0.1 and 2, to deliver different signal-to-noise ratios. 

\begin{figure}[H]
	\centering
	\includegraphics[width=14cm]{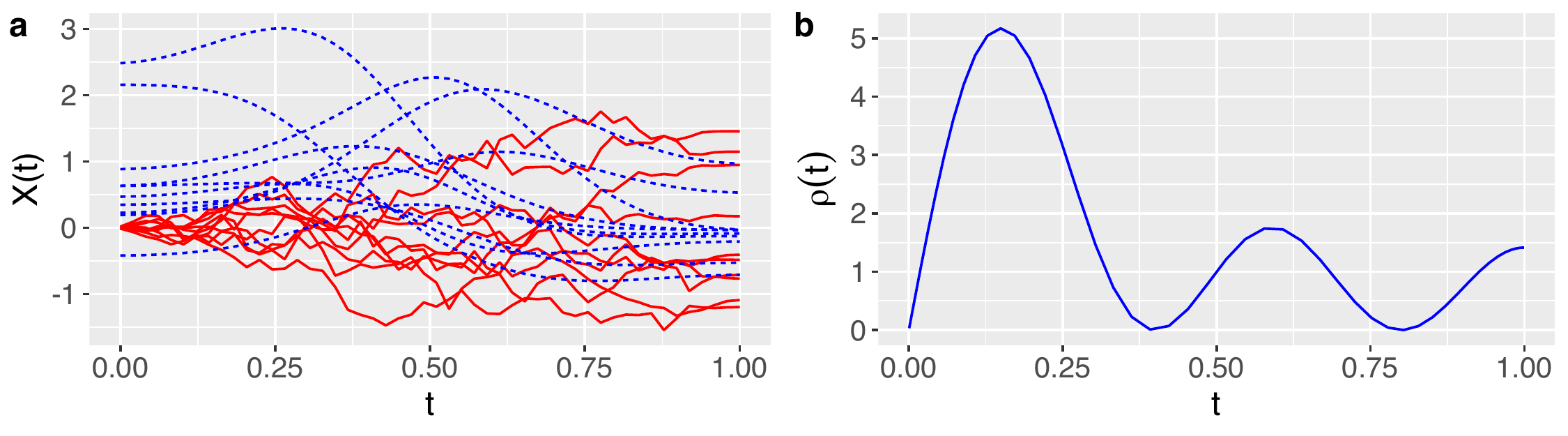}
	\caption{(a): trajectories of the functional covariate $X$ generated by the GRB kernel (blue dotted lines) and the BMC kernel (red solid lines). (b): function $\rho(t)$ in models 1 and 2.}
	\label{fig:design}
\end{figure}

In each simulation scenario, we randomly generate 100 pairs of $(X\lo i, Y\lo i)$'s as the training set and 500 pairs as the test set. 
For the two alternative estimators, the prediction error is defined as the median of the integrated
squared errors $\mbox{ISE} = \int\lo 0\hi 1 \{\hat{Y}\lo i(t) - Y\lo i(t)\}\hi 2 dt$ calculated on the test set. 
The spaces $\ca H \lo X$ and $\ca H \lo Y$ are always constructed using the same kernel: either GRB or BMC, and GRB is always used to construct $\frak M \lo X$. We leverage the GCV criteria proposed in Section \ref{sec:tuning} to choose tuning parameters in $\ca H \lo X$, $\ca H \lo Y$ and $\frak M \lo X$.
To better assess the performance of our proposed method in comparison with other methods, each simulation scenario is repeated 200 times.

\subsection{Results for dense design} \label{sec:densesimul}
In the dense design, $X \lo i$ and $Y \lo i$ are observed at 50 equally spaced time points in [0, 1]. Table \ref{tab:densepred} summarizes the medians and the inter quartiles of the prediction errors for each method across the 200 simulation runs. Our method has much better prediction accuracy than its competitors regardless of the signal-to-noise level. Moreover, even when we used the wrong kernel, for instance when $X$ is generated by the BMC kernel but we use the GRB kernel to construct both $\ca H \lo X$ and $\ca H \lo Y$, our method still achieves satisfactory prediction accuracy. This demonstrates the robustness of our method against the choice of the kernel when constructing $\ca H \lo X$ and $\ca H \lo Y$.  

We next construct the pointwise confidence intervals described in Theorem \ref{theorem:clt}. 
In particular, we randomly selected one subject from  model 2 with $\sigma = 0.1$, and constructed a confidence interval for  $\E(Y | x\lo 0)(t)$ at any $t \in [0, 1]$. Figure \ref{fig:pointPI} displays the pointwise 95\% confidence intervals. Regardless of the choice of the kernel used to generate $X$ or construct $\ca H \lo X$ in model fitting, the intervals cover the true conditional mean reasonably well. In particular, the estimated conditional mean shows a relatively large fluctuation around 0.18 due to the shape of $\rho$ shown in the right panel Figure \ref{fig:design}. After around $t = 0.25$, the magnitude of $\rho$ becomes relatively smaller; it implies smaller variability of the true conditional mean at $t > 0.25$. Consequently, the pointwise confidence intervals become considerably narrower in this region, which is consistent with the shape of $\rho$.

\begin{table}[H] 
	\tabcolsep 0.1in
	\centering
	\caption{Summary of the medians and the interquartile ranges (in parentheses) of the prediction errors across the 200 simulation runs under different simulation scenarios for each method in the dense design. The column of $X$ indicates which kernel is used to generate $X$ in model 1 or 2, and the columns of NLFFR (GRB) and NLFFR (BMC) indicate which kernel is used to construct $\ca H \lo X$ and $\ca H \lo Y$ when using the proposed NLFFR.}
	{\renewcommand{\arraystretch}{1.5}
		\begin{tabular}{@{}|c|c|c|c|c|c|c|@{}}
			\hline 
			Model & $X$ & $\sigma$ & \multicolumn{4}{c|}{Methods}  \\
			\cline{1-7}
			& & & FPCA & PLFFR& NLFFR (GRB) & NLFFR (BMC) \\
			\multirow{4}{*}{1} & \multirow{2}{*}{GRB} & 0.1 & 6.23 (1.69) & 6.77 (1.52) & 3.44 (5.07) & 1.73 (1.21)  \\
			& & 2 & 8.65 (1.70) & 9.24 (1.54) & 5.82 (5.01) & 4.13 (1.06)  \\
			&  \multirow{2}{*}{BMC} & 0.1 & 1.21 (0.23) & 1.26 (0.12) & 0.56 (0.08) & 0.43 (0.06)     \\
			& & 2 & 2.63 (0.24) & 2.69 (0.17) & 1.76 (0.12) &  1.63 (0.11) 
			\\
			\hline
			\multirow{4}{*}{2} & \multirow{2}{*}{GRB} & 0.1 & 2.12 (0.20) & 3.01 (0.23) & 0.19 (0.08) & 0.20 (0.12)   \\
			& & 2 & 3.53 (0.27) & 4.15 (0.23) & 2.21 (0.27) & 2.24 (0.29)  \\
			&  \multirow{2}{*}{BMC} & 0.1 & 1.74 (0.19) & 2.43 (0.22) & 0.47 (0.32) &  0.45 (0.30)   \\
			& & 2 & 3.31 (0.26) & 3.86 (0.26) & 2.68 (0.42) &  2.68 (0.39) \\
			\hline
	\end{tabular}}
	\label{tab:densepred}
\end{table}

\begin{figure}[H]
	\centering
	\includegraphics[width=14cm]{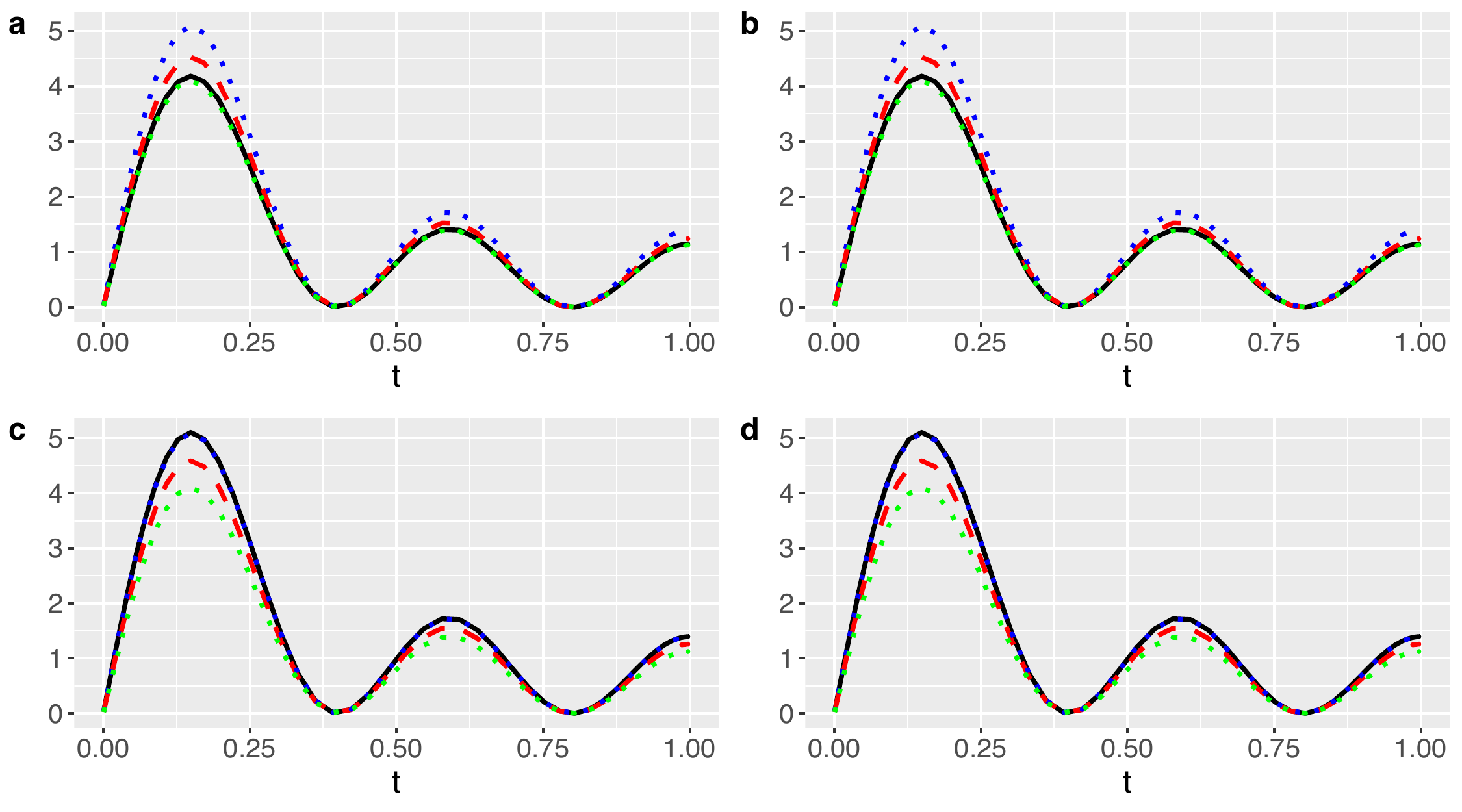}
	\caption{Pointwise confidence intervals for one randomly sampled subject in the test set from model 2. (a) \& (b): $X$ is generated by the GRB kernel, while in model fitting $\ca H \lo X$ is constructed via the the GRB and the BMC kernel, respectively. (c) \& (d): $X$ is generated by the BMC kernel, while in model fitting $\ca H \lo X$ is constructed via the the GRB and the BMC kernel, respectively.
		In each panel, the solid black line represents the true conditional mean function $\E(Y | x\lo 0)(t)$, the red dashed line represents the estimated mean, and the blue and green dotted lines represent the upper and the lower bounds of 95\% pointwise confidence intervals, respectively.}
	\label{fig:pointPI}
	
\end{figure}

We further study the simultaneous confidence band of $\E(Y\lo {n+1} | X\lo {n+1})$ given by \eqref{eq-simulCI}. Estimation of $s\lo n$ in \eqref{eq-simulCI} is straightforward. To determine the value of $C(\alpha)$, we first calculate $\hat{U}\lo i$'s on the training set based on the observed $Y\lo i$ and the estimated mean. Then a plugged-in estimate of $\Sigma\lo {UU}$ is available. We generated a large number of sample paths of a centered Gaussian process with the estimated $\Sigma\lo {UU}$ as the covariance function. Let $Z\lo i(t), i = 1, \ldots, N$ denote the randomly generated sample paths. For each of them, $\sup\lo {t \in T}|Z\lo i(t)|$ is approximated by evaluating $|Z\lo i(t)|$ on a dense grid of $T$ and then taking the maximum. The value of $C(\alpha)$ is taken as the $\alpha$-upper empirical quantile of $\sup\lo {t \in T}|Z\lo i(t)|$'s. 
Table \ref{tab:denseCP} presents the average of the true coverage probabilities of the 95\% simultaneous confidence bands across the 200 simulation runs with $\sigma = 0.1$. The true coverage probabilities for both models 1 and 2 are close to the nominal level (95\%) in most cases. Note that the coverage probability for the design when $X$ is generated from the GRB kernel in model 1 is slightly lower than the nominal level. This result is consistent with as is shown in Table \ref{tab:densepred}: compared with other designs, the prediction accuracy of the proposed method is slightly worse in this design.

\begin{table}[H] 
	\tabcolsep 0.15in
	\centering
	\caption{Summary of the mean of the coverage probability of simultaneous confidence bands over the 200 simulation runs with $\sigma = 0.1$. The first column indicates the kernel used to construct $\ca H \lo X$ in model fitting, and the row with $X:(\cdots)$ indicates the kernel used to generate $X$ in either model 1 or 2. }
	
	\begin{tabular}{@{}llcccc@{}}
		\hline 
		\multirow{3}{*}{Method}  & \multicolumn{5}{c}{Model}  \\
		\cline{2-6}
		& \multicolumn{2}{c}{1} & & \multicolumn{2}{c}{2} \\
		\cline{2-3} \cline{5-6}
		& $X$: GRB & $X$: BMC & & $X$: GRB & $X$: BMC\\
		
		\hline
		
		GRB  & 0.896 & 0.976 & & 0.924 & 0.946   \\
		
		BMC & 0.898 & 0.972 & & 0.916 & 0.962    \\

		\hline
	\end{tabular}
	\label{tab:denseCP}
\end{table}

\subsection{Results for sparse design}
In the sparse design, $X \lo i$ and $Y \lo i$ are observed at 10 time points on [0, 1], randomly selected from the 50 equally spaced time points in the dense design. For the two alternative methods, we employ the principal component analysis through conditional expectation (PACE) method proposed by \cite{yao2005a} to recover each sparse trajectory first. Our method can still fit such data without any extra pre-processing. Prediction errors on the test set fitted by each method are summarized in Table \ref{tab:sprasepred}. In comparison with the dense case, our proposed method displays a similar advantage over the two competitors in terms of prediction accuracy.

\begin{table}[H] 
	\tabcolsep 0.1in
	\centering
	\caption{Summary of the medians and the interquartile ranges (in parentheses) of the prediction errors across the 200 simulation runs under different simulation scenarios for each method in the sparse design. The column of $X$ indicates which kernel is used to generate $X$ in model 1 or 2, and the columns of NLFFR (GRB) and NLFFR (BMC) indicate which kernel is used to construct $\ca H \lo X$ and $\ca H \lo Y$.} 
	{\renewcommand{\arraystretch}{1.5}
		\begin{tabular}{@{}|c|c|c|c|c|c|c|@{}}
			\hline 
			Model & $X$ & $\sigma$ & \multicolumn{4}{c|}{Methods}  \\
			\cline{1-7}
			& & & FPCA & PLFFR& NLFFR (GRB) & NLFFR (BMC) \\
			\multirow{4}{*}{1} & \multirow{2}{*}{GRB} & 0.1 & 6.35 (1.67) & 6.97 (2.73) & 3.48 (2.61) & 2.17 (1.14)  \\
			& & 2 & 9.12 (2.16) & 9.60 (3.11) & 5.92 (5.49) & 4.66 (1.28)  \\
			&  \multirow{2}{*}{BMC} & 0.1 & 1.10 (0.24) & 1.42 (0.27) & 0.65 (0.10) & 0.54 (0.09)     \\
			& & 2 & 2.54 (0.28) & 2.88 (0.27) & 1.89 (0.14) &  1.82 (0.16) 
			\\
			\hline
			\multirow{4}{*}{2} & \multirow{2}{*}{GRB} & 0.1 & 2.11 (0.22) & 3.01 (0.31) & 0.21 (0.07) & 0.20 (0.08)   \\
			& & 2 & 3.50 (0.29) & 4.12 (0.30) & 2.17 (0.24) & 2.16 (0.25)  \\
			&  \multirow{2}{*}{BMC} & 0.1 & 1.77 (0.18) & 2.46 (0.25) & 0.43 (0.18) &  0.46 (0.19)   \\
			& & 2 & 3.21 (0.25) & 3.83 (0.28) & 2.60 (0.31) &  2.62 (0.32) \\
			\hline
	\end{tabular}}
	\label{tab:sprasepred}
\end{table}

\section{Data application} \label{sec-real}
In this section, we apply our proposed method and the aforementioned competitors to a data application. We are not only interested in predication accuracy of our method in real applications, but also the pointwise confidence intervals and the simultaneous confidence band introduced in 
Section \ref{sec-CLT}.

\begin{figure}[H]
	\centering
	\includegraphics[width=14cm]{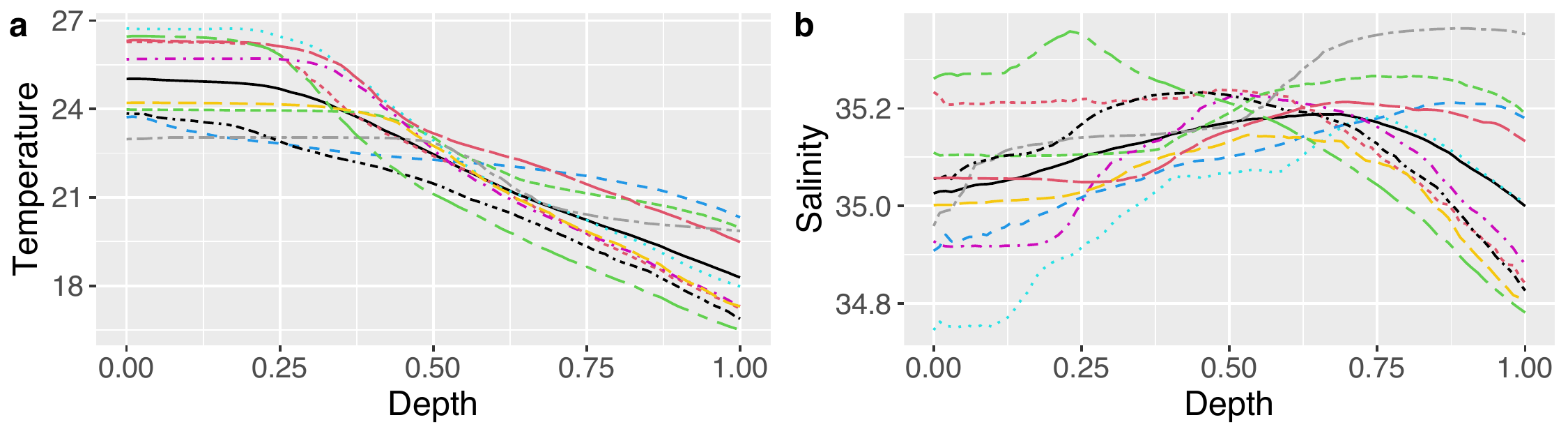}
	\caption{(a): ten sample curves of Temperature and the sample mean curve of them (black solid line). (b):  ten sample curves of Salinity and the sample mean curve of them (black solid line). }
	\label{fig:ocean}
\end{figure}

As indicated by the website (http://hahana.soest.
	hawaii.edu/hot/hot-dogs/cextraction.html), the Hawaii Ocean Time-series (HOT) program has been collecting time course observations on the hydrography, chemistry and biology of the water column at a station north of Oahu, Hawaii since October 1988. One goal of this program is to learn about concentrations of some materials in the upper water column (0 - 200 m below the sea surface). With the aid of CTD sampling support, profiles of temperature, salinity, oxygen and potential density as a function of pressure (or equivalently depth) are available. In our study, we took a portion of the whole data set. The data set has five variables: Temperature, Salinity, Potential Density, Oxygen and Chloropigment and, in a single day, each of them has 101 measurements, one per two meters from 0 to 200 meters. They can be treated as a function of depth, and trajectories collected from different days are viewed as different sample curves. There are 116 sample curves in total for each variable. 

In this study, we are interested in using the trajectories of Temperature to predict those of Salinity. 
As indicated by \cite{good2013}, Temperature is strongly associated with Salinity and there exists a nonlinear relationship between them. This assertion can be further justified by Figure \ref{fig:ocean}, which shows the trajectories of Temperature and Salinity of 10 randomly selected samples, where the depth was rescaled to [0, 1] from [0, 200]. The trends of these two groups of mean curves suggest that Temperature decreases as the depth increases, whereas Salinity goes up first and then drops down as depth increases. Additionally, the response variable, Salinity, displays more variability near the boundary than in the interior region.

To evaluate prediction accuracy of each method, we randomly and evenly split the whole data set into a training set and a test set. 
Each method was fitted to the training set and then the fitted function-on-function regresion was used to predict the response in the test set. This process was repeated $M = 200$ times to assess variability in predictions. The medians and the interquartile ranges
of the prediction errors across the 200 splits are shown in Table \ref{tab:Ocean}.
Our proposed method greatly outperforms the two competitors and the poor performances of the FPCA and PLFFR methods indicate that the relationship between Salinity and Temperature cannot be adequately fitted by a linear function-one-function regression model. 

\begin{table}[H] 
	\centering	
	\caption{Summary of the averages and standard deviations (in parentheses) of the prediction errors across the 200 random splits.}
	\begin{tabular}{|c|c|c|c|c|}
		\hline
		& FPCA & PLFFR & NLFFR (GRB) & NLFFR(BMC)  \\
		\hline
		median	&  $1.24 \times 10 \hi 3$ & $1.24 \times 10 \hi 3$  & $1.20 \times 10 \hi {-2}$ & $1.22 \times 10 \hi {-2}$  \\
		\hline
		IQR	 & (0.51) & (38.78) & ($1.48 \times 10 \hi {-3}$) &  ($1.48 \times 10 \hi {-5}$) \\
		\hline
	\end{tabular} 
	\vspace*{-5pt}
	\label{tab:Ocean}
\end{table}

\begin{figure}[H]
	\centering
	\includegraphics[width=14cm]{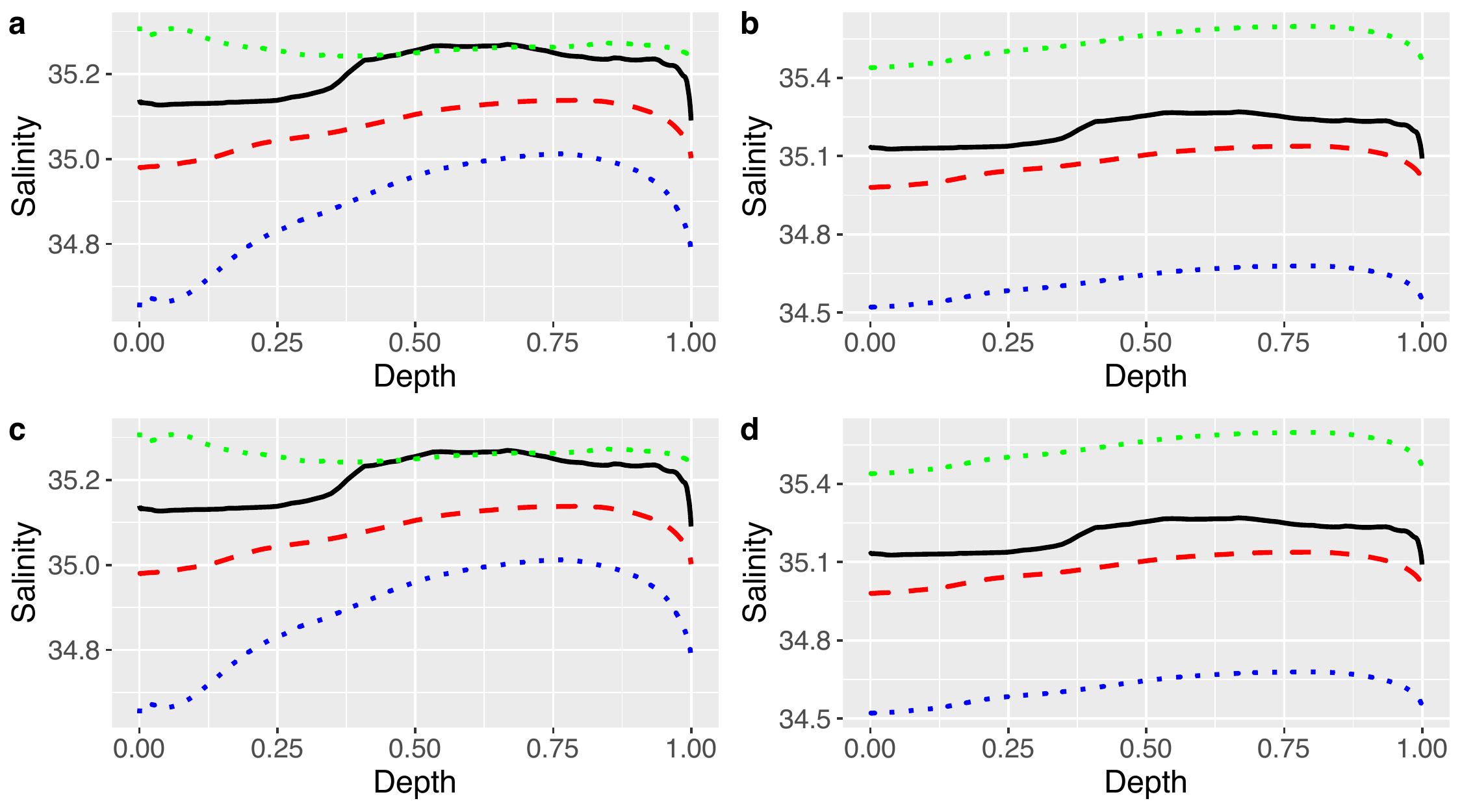}
	\caption{(a) \& (b): Pointwise confidence intervals and simultaneous confidence bands constructed from the GRB kernel for one random sample in the test set. (c) \& (d): Pointwise confidence intervals and simultaneous confidence bands constructed from the BMC kernel for one random sample in the test set. In each panel, the solid black line represents the observed trajectory of $Y$, the red dashed line represents the estimated mean, and the blue and green dotted lines represent the upper and the lower bounds of 95\% pointwise (or simulataneous) confidence intervals, respectively. }
	\label{fig:oceanCI}
\end{figure}

We also constructed the pointwise confidence intervals defined by Theorem \ref{theorem:clt} and the simultaneous confidence band by \eqref{eq-simulCI} for this regression problem. Figure \ref{fig:oceanCI} shows both the 95\% pointwise confidence intervals and the 95\% simultaneous confidence band constructed by the GRB and the BMC kernels for one randomly selected sample from the test set. The shapes of both the pointwise confidence intervals and the simultaneous bands are similar under these two kernels. It implies that pointwise confidence intervals and simultaneous confidence bands are robust to the choice of the kernel used to construct $\ca H \lo X$ and $\ca H \lo Y$. Not surprisingly, the simultaneous confidence band is wider than the pointwise confidence intervals for both kernels. Furthermore, the two left panels of Figure \ref{fig:oceanCI} indicate that the predicted mean response tends to be more variable near the boundary in comparison with the interior region. This finding is consistent with what we have seen from the right panel of Figure \ref{fig:ocean}.

\section{Conclusions} \label{sec-conc}

In this paper we have proposed a nonlinear function-on-function regression model based on a linear operator in RKHS. Compared with the current linear function-on-function regression approaches, our approach shows a remarkable improvement in prediction accuracy. In addition, with the aid of nested Hilbert spaces, our method avoids the large number of parameters that need to be estimated when the tensor products of spline basis functions or the eigenfunctions of the predictor and response are deployed in linear function-on-function regression (\citealp{ramsay2005}; \citealp{yao2005b}; \citealp{sun2018optimal}). The estimation procedure can accommodate irregularly and sparsely observed functional predictor and response.

Existing asymptotic development on function-on-function regression was focused on consistency and convergence rates. For instance, \cite{sun2018optimal} studied the minimax rate in mean prediction using an RKHS-based approach. Both consistency and the convergence rate were established by \cite{luo2017} in a linear function-on-function regression model. However, little work has been done to develop statistical inferences for function-on-function regression. Though there were several precursors in this regard [see \cite{yao2005b} and \cite{crambes2013} for example], they were mainly concerned with linear models. In comparison, our theoretical development includes both convergence rate, pointwise and uniform central limit theorem of the regression estimate.

\setcounter{equation}{0}
\setcounter{figure}{0}
\setcounter{table}{0}
\makeatletter
\renewcommand{\theequation}{A\arabic{equation}}
\renewcommand{\thefigure}{A\arabic{figure}}
\renewcommand{\bibnumfmt}[1]{[A#1]}
\renewcommand{\citenumfont}[1]{A#1}



\section*{Appendix}
In this section we provide the proofs of the theorems, lemmas, and corollaries in the manuscript. The equation labels such as (1) and (2) are for the equations in the manuscript; equation labels such as (A1) and (A2) are for the equations in this appendix.

\begin{proof}[Proof of Theorem \ref{th-sol}]
	Denote $g - \E (\langle g, Y \rangle \lo {\ca H \lo Y} )$ and $(B g) - \E[(Bg)(X)]$ by $g\lo c$ and $(Bg)\lo c$, respectively. 
	Note that 
	\begin{align*}
		\ali \quad \E [ \{\langle g\lo c, Y \rangle \lo {\ca H \lo Y}  - (B g)\lo c(X)  \}  \hi 2 ] \\
		\ali = \E [ \{ \langle g\lo c, Y \rangle \lo {\ca H \lo Y} - (B \lo 0  g)\lo c(X) + (B \lo 0  g)\lo c(X) - (B  g)\lo c(X)  \} \hi 2 ] \\
		\ali = \E [ \{ \langle g\lo c, Y \rangle \lo {\ca H \lo Y} - (B \lo 0  g)\lo c(X)  \} \hi 2 ]
		+2  \E \{[( \langle g\lo c, Y \rangle \lo {\ca H \lo Y} - (B \lo 0  g)\lo c(X)][ (B \lo 0  g)\lo c(X) - (B  g)\lo c(X) ] \} \\
		\ali \quad+   \E [ \{ (B \lo 0  g)\lo c(X) - (B  g)\lo c(X)  \} \hi 2 ]. 
	\end{align*}
	The cross product term is
	\begin{align*}
		\ali \quad\E \{[\langle g\lo c, Y \rangle \lo {\ca H \lo Y} - (B \lo 0  g)\lo c(X)][ (B \lo 0  g)\lo c(X) - (B  g)\lo c(X) ] \} \\
		\ali =  \E \{( \langle g\lo c, Y \rangle \lo {\ca H \lo Y}) (B \lo 0  g)\lo c(X)  \}   - \E \{( \langle g\lo c, Y \rangle \lo {\ca H \lo Y})  (B  g)\lo c(X)  \} \\
		\ali\quad - \E \{ (B \lo 0  g)\lo c(X) (B \lo 0  g)\lo c(X)  \}  + \E \{ (B \lo 0  g)\lo c(X)  (B  g)\lo c(X)  \} \\
		\ali =  \langle g, \Sigma \lo {YX} B \lo 0 g \rangle \lo {\ca H \lo Y}   - \langle g, \Sigma \lo {YX} B g \rangle \lo {\ca H \lo Y} - \langle B \lo 0 g, \Sigma \lo {XX} B \lo 0 g \rangle \lo {\frak M \lo X}    + \langle B \lo 0 g, \Sigma \lo {XX} B g \rangle \lo {\frak M \lo X}  \\
		\ali =  \langle g, \Sigma \lo {YX} \Sigma \lo {XX} \mpinv \Sigma \lo {XY} g \rangle \lo {\ca H \lo Y}  - \langle g, \Sigma \lo {YX} B g \rangle \lo {\ca H \lo Y} 
		- \langle \Sigma \lo {XX} \mpinv \Sigma \lo {XY}  g, \Sigma \lo {XX} \Sigma \lo {XX} \mpinv \Sigma \lo {XY}  g \rangle \lo {\frak M \lo X}  \\
		\ali\quad + \langle \Sigma \lo {XX} \mpinv \Sigma \lo {XY}  g, \Sigma \lo {XX} B g \rangle \lo {\frak M \lo X} \\
		\ali = 0,
	\end{align*} 
	where the last equality holds since $\Sigma \lo {XX}\Sigma \lo {XX} \mpinv$ is an identity mapping from $\ran(\Sigma \lo {XX})$ onto $\ran(\Sigma \lo {XX})$. Therefore, 
	\begin{align*}
		\ali \quad\E [ \{ \langle g\lo c, Y \rangle \lo {\ca H \lo Y} - (B g)\lo c(X)  \} \hi 2 ] \\
		\ali = \E [ \{ \langle g\lo c, Y \rangle \lo {\ca H \lo Y}- (B \lo 0  g)\lo c(X)  \} \hi 2 ]
		+   \E [ \{ (B \lo 0  g)\lo c(X) - (B  g)\lo c(X)  \} \hi 2 ] \\
		\ali \ge \E [ \{ \langle g\lo c, Y \rangle \lo {\ca H \lo Y} - (B \lo 0  g)\lo c(X)  \} \hi 2 ]
	\end{align*}
	as desired. 
\end{proof}

\begin{proof}[Proof of Proposition \ref{prop-reg}]
	Take an arbitrary $h \in \frak M \lo X$. Then we have
	\begin{align*}
		\ali \cov[\langle g, Y \rangle \lo {\ca H \lo Y} - (B\lo 0g)(X), h(X) ] \\
		\ali = \cov[\langle g, Y \rangle \lo {\ca H \lo Y}, h(X)] - \cov[(B\lo 0g)(X), h(X)]\\
		\ali = \langle h, \Sigma \lo {XY} g \rangle \lo {\frak M \lo X} - \langle \Sigma \lo {XX} (B\lo 0g), h \rangle \lo {\frak M \lo X} \\
		\ali = \langle h, \Sigma \lo {XY} g \rangle \lo {\frak M \lo X} - \langle \Sigma \lo {XX} \Sigma \lo {XX} \mpinv \Sigma \lo {XY}g, h \rangle \lo {\frak M \lo X} \\
		\ali = 0, 
	\end{align*}
	where the last equation holds since $\Sigma \lo {XX}\Sigma \lo {XX} \mpinv$ is an identity mapping from $\ran(\Sigma \lo {XX})$ onto $\ran(\Sigma \lo {XX})$. 
	Since $h$ is an arbitrary function chosen from $\frak M \lo X$ and $\frak M \lo X + \R$ is dense in $L\lo 2(P\lo X)$, 
	$\E[\langle g, Y \rangle \lo {\ca H \lo Y} | X] - (B\lo 0g)(X)$ must be a constant. Taking an unconditional expectation leads to \eqref{eq-reg}. 
\end{proof}

\begin{proof}[Proof of Proposition \ref{prop-condE}]
	Because $\Sigma\lo{YX}\Sigma\lo{XX}\mpinv$ is bounded, its domain can be extended 
	from $\ran(\Sigma\lo{XX})$ to $\cran(\Sigma\lo{XX})$, which is $\frak M \lo X \hi 0$. 
	Therefore, we take the domain of $B\lo 0\hi *$ as $\frak M \lo X \hi 0$. 
	To show \eqref{eq-condE}, it suffices to verify that for any $h \in \ca H \lo Y$, 
	\begin{equation}
		\langle \E(Y\lo c | X), h \rangle \lo {\ca H \lo Y} = \langle B\lo 0 \hi * \ka\lo c(\cdot, X), h\rangle \lo {\ca H \lo Y}.
		\label{eq-condrep}
	\end{equation}
	Obviously, the left-hand side is $\E[\langle h, Y\lo c \rangle \lo {\ca H \lo Y}| X]$, which, by Proposition \ref{prop-reg}, is equal to $(B\lo 0h)(X) - \E[(B\lo 0h)(X)]$. 
	The right-hand side of \eqref{eq-condrep} is 
	\begin{align*}
		\langle B\lo 0 \hi * \ka\lo c(\cdot, X), h\rangle \lo {\ca H \lo Y}
		\ali = \langle \ka\lo c(\cdot, X), B\lo 0h \rangle \lo {\frak M \lo X} \\
		\ali = (B\lo 0h)(X) - \E[(B\lo 0h)(X)],
	\end{align*}
	which agrees with the  left-hand side of \eqref{eq-condrep}. 
\end{proof}

\begin{proof}[Proof of Lemma \ref{lemma:reg and res}]
	{\em 1}. Since, for any $f \lo 1, f \lo 2 \in  \frak M \lo X$, $f \lo 1 \otimes (B \lo 0 \hi * f \lo 2) = (f \lo 1\otimes f \lo 2) B \lo 0$, we have
	\begin{align*}
		\Sigma \lo {XU} = \ali \E [ (\ka (\cdot, X) - \mu \lo X) \otimes \{Y  - \mu \lo Y - B \lo 0 \hi * (\ka (\cdot, X) - \mu \lo X)\}] \\
		= \ali  \E [ (\ka (\cdot, X) - \mu \lo X) \otimes  (Y - \mu \lo Y)  ] - \E [ (\ka (\cdot, X) - \mu \lo X) \otimes (B \lo 0 \hi * (\ka (\cdot, X) - \mu \lo X))] \\
		= \ali \Sigma \lo {XY} - \Sigma \lo {XX} B \lo 0 = 0.
	\end{align*}
	
	\noindent {\em 2}. By definition,
	\begin{align*}
		\hat \Sigma \lo {XU} = \ali \E \lo n
		[ (\ka (\cdot, X) - \hat \mu \lo X) \otimes \{Y - \mu \lo Y - B \lo 0 \hi * (\ka (\cdot, X) - \mu \lo X)\}] \\
		= \ali
		\E \lo n
		[ (\ka (\cdot, X) - \hat \mu \lo X) \otimes  (Y - \mu \lo Y) ]
		- \E \lo n
		[ (\ka (\cdot, X) - \hat \mu \lo X) \otimes \{B \lo 0 \hi * (\ka (\cdot, X) - \mu \lo X)\} ].
	\end{align*}
	The first term on the right is $\hat \Sigma \lo {XY}$. Since $\E \lo n
	[\ka (\cdot, X) - \hat{\mu} \lo X] = 0$, the second term is unchanged if we replace $\mu \lo X$ in $\ka (\cdot, X) - \mu \lo X$ by $\hat \mu \lo X$. Thus it can be rewritten as
	\begin{align*}
		\E \lo n
		[ (\ka (\cdot, X) - \hat \mu \lo X) \otimes (B \lo 0 \hi * (\ka (\cdot, X) - \hat \mu \lo X)) ] = \hat{\Sigma} \lo {XX} B \lo 0,
	\end{align*}
	as desired.  
\end{proof}

\begin{proof}[Proof of Lemma \ref{lemma:hat vs tilde}]
	By the definitions of $\hat \Sigma \lo {XU}$ and $\tilde \Sigma \lo {XU}$ and some simple calculation, we have
	\begin{align}\label{eq:difference}
		\hat \Sigma \lo {XU} - \tilde \Sigma \lo {XU} = (\hat \mu \lo X - \mu \lo X) \otimes \hat \mu \lo U.
	\end{align}
	Hence
	\begin{align*}
		\| \hat \Sigma \lo {XU} - \tilde \Sigma \lo {XU} \| \loo{HS} = \| (\hat \mu \lo X - \mu \lo X ) \otimes \hat \mu \lo U \| \loo {HS} = \| \hat \mu \lo X - \mu \lo X \| \lo {\frak M \lo X } \, \| \hat \mu \lo U \| \lo {\ca H \lo Y}.
	\end{align*}
	By Chebychev's inequality, it can be easily shown that $\| \hat \mu \lo X - \mu \lo X \| \lo {\frak M \lo X } = O \lo P(n \hi {-1/2})$ and $\| \hat \mu \lo U \| \lo {\ca H \lo Y} = O \lo P (n \hi {-1/2})$, which imply the asserted result. 
\end{proof}

\begin{proof} [Proof of Lemma \ref{lemma:calculus}]
	Let $m \lo n = \lfloor \epsilon \lo n \hi {-1/\alpha} \rfloor$. Then, by Assumption \ref{assumption:KL for kappa}({ii}), 
	\begin{align}\label{eq:integral approximation}
		\begin{split}
			\tsum \lo {j \in \nat}(\lambda \lo i + \epsilon \lo n ) \hi {-2} \lambda  \lo j
			\le \ali \tsum \lo {j = 1} \hi { m \lo n} \lambda \lo i \hi {-1}
			+
			\epsilon \lo n \hi {-2} \tsum \lo {j = m \lo n + 1} \hi \infty \lambda  \lo j  \\
			\asymp \ali \int \lo 1 \hi {m \lo n} x \hi \alpha dx
			+
			\epsilon \lo n \hi {-2} \int \lo {m \lo n} \hi \infty x \hi {-\alpha } d x   \asymp \epsilon \lo n \hi {-(\alpha+1)/\alpha},
		\end{split}
	\end{align}
	as desired.
\end{proof}

\begin{proof}[Proof of Theorem \ref{theorem:convergence rate}]
	{\em 1}. Using Lemma \ref{lemma:reg and res}, we decompose $\hat B$ as $\hat B \loo {reg} + \hat B \loo {res}$, where
	\begin{align*}
		\hat B \loo {reg} = \hat V \hat \Sigma \lo {XX} B \lo 0, \quad \hat B \loo {res} = \hat V \hat \Sigma \lo {XU}.
	\end{align*}
	As suggested by the notation, $\hat B \loo {reg}$ represents the regression part of $\hat B$, whereas $\hat B \loo {res}$ the residual part.
	Let $B \lo n = V \lo n \Sigma \lo {XY}$,  which represents the population-level approximation of $B \lo 0$ via Tychonoff regularization, and further decompose $\hat B \loo {reg}$ as $\hat B \loo {reg}- B \lo n + B \lo n$. We have
	\begin{align*}
		\hat B - B \lo 0
		=   \hat B \loo {res} +   (\hat B \loo {reg} - B \lo n)  +   ( B \lo n - B \lo 0 ).
	\end{align*}

	We first analyze the regression term $\hat B \loo {reg} - B \lo n$. By construction,
	\begin{align*}
		\hat B \loo {reg} - B \lo n=\hat V \hat \Sigma \lo {XX} B \lo 0 - B \lo n = \hat V \hat \Sigma \lo {XX} B \lo 0 - V \lo n \Sigma \lo {XX} B \lo 0 = (\hat V \hat  \Sigma \lo {XX} - V \lo n \Sigma  \lo {XX} ) B \lo 0.
	\end{align*}
	Since $\hat V$ and $\hat \Sigma \lo {XX}$ commute, and $V \lo n$ and $\Sigma \lo {XX}$ commute, we can rewrite
	\begin{align*}
		\hat V \hat  \Sigma \lo {XX} -\Sigma  \lo {XX}  V \lo n = \ali \hat V (\hat \Sigma \lo {XX} V \lo n \inv - \hat V \inv  \Sigma \lo {XX} ) V \lo n = \epsilon \lo n \hat V (\hat \Sigma \lo {XX} - \Sigma \lo {XX} ) V \lo n.
	\end{align*}
	Therefore,
	\begin{align*}
		\|\hat B \loo {reg} - B \lo n \| \loo {OP}
		\le \ali \|  \epsilon \lo n \hat V \| \loo {OP} \, \|\hat \Sigma \lo {XX} - \Sigma \lo {XX} \| \loo {OP} \, \| V \lo n B \lo 0\| \loo {OP}   = O \lo P ( n \hi {-1/2} )  \, \| V \lo n B \lo 0 \| \loo {OP},
	\end{align*}
	where the second equality holds because  $\|\hat \Sigma \lo {XX} - \Sigma \lo {XX} \| \loo {OP} = O \lo P (n \hi {-1/2})$ (Lemma \ref{lemma:fuku et al}) and $\| \epsilon \lo n \hat V \| \loo {OP} \le \| I \| \loo {OP} = 1$. By Assumption \ref{assumption:finte moments and beta}(ii),   $  V \lo n B \lo 0  = V \lo n V \Sigma \lo {XX}  \hi {1+\beta} S \lo {XY}$ for a bounded operator $S \lo {XY}$. Hence if $ \beta \in (0, 1]$, 
	\begin{align*}
		\nonumber
		\| V \lo n B \lo 0 \| \loo {OP} & = \| V \lo n V \Sigma \lo {XX}  \hi {1+\beta} S \lo {XY} \| \loo {OP} = \| V \lo n \Sigma \lo {XX} \hi {\beta} \| \loo {OP} \| S \lo {XY}\| \loo {OP} \\
		\nonumber 
		& \leq \|(\Sigma\lo{XX} + \epsilon\lo n I) \hi {-1 + \beta}\| \loo{OP} \| S \lo {XY}\| \loo {OP} \\
		\nonumber
		& = \epsilon \lo n \hi {\beta - 1} \|(\Sigma\lo{XX} + \epsilon\lo n I) \hi {-1 + \beta} (\epsilon\lo n I) \hi {1 - \beta}\|\loo{OP} \| S \lo {XY}\| \loo {OP} \\
		&	= O (\epsilon \lo n \hi {\beta - 1}).
	\end{align*}
	If $\beta > 1$, one has
	\begin{align*}
		\| V \lo n \Sigma \lo {XX} \hi {\beta} \| \loo {OP}  \leq  \|(\Sigma\lo{XX} + \epsilon\lo n I) \inv \Sigma\lo{XX}\| \loo{OP} \|\Sigma\lo{XX}\|\loo{OP} = O(1).
	\end{align*}
	It follows that
	\begin{equation} \label{eq:Vn B0}
		\| V \lo n B \lo 0 \| \loo {OP} = O(\epsilon \lo n \hi {\beta \wedge 1 - 1}).
	\end{equation}
	Consequently,
	\begin{align}\label{eq:A n2}
		\|\hat B \loo {reg} - B \lo n \| \loo {OP}= O \lo P ( n \hi {-1/2}\epsilon \lo n \hi {  \beta \wedge 1 - 1} ).
	\end{align}

	Secondly, we analyze the bias term $ B \lo n - B \lo 0 $. Since
	\begin{align*}
		( V \lo n - V  )\Sigma \lo {XY} = V \lo n  [ \Sigma \lo {XX} - (\Sigma \lo {XX}  + \epsilon \lo n  I ) ] V  \Sigma \lo {XY}= -  \epsilon \lo n  V \lo n B \lo 0,
	\end{align*}
	we have, by (\ref{eq:Vn B0}),
	\begin{align}\label{eq:A n3}
		\|B \lo n - B \lo 0 \| \loo {OP}= O (\epsilon \hi {\beta \wedge 1}).
	\end{align}

	Thirdly, we analyze  $\hat B \loo {res}$, which can be further decomposed as
	\begin{align*}
		(\hat V \hat \Sigma \lo {XU} - \hat V \tilde \Sigma \lo {XU}  )  + (\hat V \tilde \Sigma \lo {XU} - V \lo n \tilde  \Sigma \lo {XU} ) +   V \lo n \tilde  \Sigma \lo {XU}.
	\end{align*}
	Since $\| \hat V \| \loo {OP} \le \epsilon \lo n \inv \| I \| \loo {OP} = \epsilon \lo n \inv$ and, by Lemma \ref{lemma:hat vs tilde},   $\| \hat \Sigma \lo {XU} - \tilde \Sigma \lo {XU} \| \loo {OP} = O \lo P ( n \inv)$, we have
	\begin{align}\label{eq:A n11}
		\|\hat V \hat \Sigma \lo {XU} - \hat V \tilde \Sigma \lo {XU} \| \loo {OP} =O \lo P ( n \hi {-1} \epsilon \lo n \inv).
	\end{align}
	Since $\hat V - V \lo n = \hat V ( \Sigma \lo {XX} - \hat \Sigma \lo {XX} ) V \lo n$, we have
	\begin{align}\label{eq:to be combined}
		\begin{split}
			\|\hat V \tilde \Sigma \lo {XU} - V \lo n \tilde  \Sigma \lo {XU}  \| \loo {OP} \le \ali  \| \hat V \| \loo {OP} \, \| \Sigma \lo {XX} - \hat \Sigma \lo {XX} \| \loo {OP} \, \| V \lo n \tilde \Sigma \lo {XU}   \| \loo {OP} \\
			= \ali  O \lo P ( n \hi {-1/2} \epsilon \lo n \inv ) \, \| V \lo n \tilde \Sigma \lo {XU}\| \loo {OP}.
		\end{split}
	\end{align}
	The term $\| V \lo n \tilde \Sigma \lo {XU} \| \loo {OP}$ is bounded by  $\|  V \lo n \tilde \Sigma \lo {XU} \| \loo {HS}$, whose square is
	\begin{align*}
		\|  V \lo n \tilde \Sigma \lo {XU} \| \loo {HS} \hi 2=\ali  \|  n \inv \tsum \lo {i=1} \hi n V \lo n [(\ka (\cdot, X \lo i) - \mu \lo X) \otimes U \lo i] \| \loo {HS} \hi 2.
	\end{align*}
	Since $\E[(\ka (\cdot, X \lo i) - \mu \lo X) \otimes U \lo i]=\Sigma \lo {XU}=0$ and $(X \lo 1, U \lo 1), \ldots, (X \lo n, U \lo n)$ are i.i.d., we have
	\begin{align*}
		\E ( \|  V \lo n \tilde \Sigma \lo {XU} \| \loo {HS} \hi 2 )
		=\ali   n \hi {-2} \tsum \lo {a=1} \hi n \tsum \lo {b=1} \hi n
		\E \left(\langle   V \lo n [(\ka (\cdot, X \lo a  ) - \mu \lo X) \otimes U  \lo a ], \,   V \lo n [(\ka (\cdot, X \lo b  ) - \mu \lo X) \otimes U  \lo b] \rangle \loo {HS} \hi 2 \right) \\
		=\ali   n \hi {-2} \tsum \lo {a=1} \hi n
		\E \left(\|   V \lo n [(\ka (\cdot, X \lo a  ) - \mu \lo X) \otimes U \lo a \| \loo {HS} \hi 2 \right) \\
		=\ali   n \hi {-1}
		\E \left(\|   V \lo n [(\ka (\cdot, X   ) - \mu \lo X) \otimes U  \| \loo {HS} \hi 2 \right).
	\end{align*}
	The squared Hilbert-Schmidt norm on the right-hand side is
	\begin{align*}
		\ali \|   V \lo n [(\ka (\cdot, X   ) - \mu \lo X) \otimes U  \| \loo {HS} \hi 2 \\
		\ali \hspace{1in} =\tsum \lo {j \in \nat} \,  \langle \varphi \lo j, \| U \| \hi 2 \langle \varphi \lo j, [ V \lo n ( \ka (\cdot, X) - \mu \lo X)] \otimes [V \lo n ( \ka (\cdot, X) - \mu \lo X )] \varphi \lo j \rangle \lo {\frak M \lo X} \\
		\ali \hspace{1in} = \| U \| \hi 2 \, \tsum \lo {j \in \nat} \, (\lambda \lo j + \epsilon \lo n ) \hi {-2}  \zeta \lo j \hi 2,
	\end{align*}
	where, for the last equality, we have used the expansion (\ref{eq:kappa expansion}).
	Taking expectation on both sides, and invoking the condition  $X \indep U$ (Assumption \ref{assumption:KL for kappa}(i)), we have
	\begin{align*}
		\E (\|   V \lo n [(\ka (\cdot, X   ) - \mu \lo X) \otimes U  \| \loo {HS} \hi 2) = \E( \| U \| \hi 2 )   \tsum \lo {j \in \nat} \, (\lambda \lo j + \epsilon \lo n ) \hi {-2} \lambda \lo j.
	\end{align*}
	By Lemma \ref{lemma:calculus}, the right hand side is of the order  $O(\epsilon \lo n \hi {-(\alpha+1) / \alpha})$. Hence $\E \| V \lo n \tilde \Sigma \lo {XU} \| \loo {HS} \hi 2$ is of the order $O (n \inv \epsilon \lo n \hi {-(\alpha+1)/\alpha})$, which, by Chebychev's inequality, implies
	\begin{align*}
		\| V \lo n \tilde \Sigma \lo {XU} \| \loo {HS} = O \lo P ( n \hi {-1/2} \epsilon \lo n \hi {-(\alpha+1)/(2\alpha)}) \,  \Rightarrow \, \| V \lo n \tilde \Sigma \lo {XU} \| \loo {OP} = O \lo P ( n \hi {-1/2} \epsilon \lo n \hi {-(\alpha+1)/(2\alpha)}).
	\end{align*}
	Combining this with (\ref{eq:to be combined}) we have
	\begin{align*}
		\|\hat V \tilde \Sigma \lo {XU} - V \lo n \tilde  \Sigma \lo {XU}  \| \loo {OP} =  O \lo P ( n \hi {-1/2} \epsilon \lo n \inv n \hi {-1/2} \epsilon \lo n \hi {-(\alpha+1)/(2\alpha)}) = O \lo P ( n \inv \epsilon \lo n \hi {- (3 \alpha + 1)/(2 \alpha)} ).
	\end{align*}
	So
	\begin{align}\label{eq:B res}
		\begin{split}
			\hat B \loo {res} = \ali   O \lo P ( n \hi {-1} \epsilon \lo n \inv + n \inv \epsilon \lo n \hi {- (3 \alpha + 1)/(2 \alpha)} + n \hi {-1/2} \epsilon \lo n \hi {-(\alpha+1)/(2\alpha)}) \\
			= \ali   O \lo P (n \inv \epsilon \lo n \hi {- (3 \alpha + 1)/(2 \alpha)} + n \hi {-1/2} \epsilon \lo n \hi {-(\alpha+1)/(2\alpha)}).
		\end{split}
	\end{align}
	Combining (\ref{eq:A n2}), (\ref{eq:A n3}),  and  (\ref{eq:B res}), we have (\ref{eq:goal 1}).
	
	{\em 2}. For the right-hand side of (\ref{eq:goal 1}) to go to 0 we need
	\begin{align*}
		n \hi {-1/2} \epsilon \lo n \hi {\beta\wedge 1 - 1} \prec 1, \quad    n\hi {-1} \epsilon \lo n \hi {- (3 \alpha+ 1)/ (2 \alpha)} \prec 1, \quad  n \hi {-1/2} \epsilon \lo  n \hi {-(\alpha + 1)/(2 \alpha)} ) \prec 1,
	\end{align*}
	which are satisfied if
	\begin{align*}
		\epsilon \lo n \succ n \hi {-1/[2 \{1-(\beta \wedge 1)\}]}, \quad \epsilon \lo n \succ n \hi {-(2 \alpha)/(3 \alpha+1)}, \quad \epsilon \lo n \succ n \hi {- \alpha/ (\alpha + 1)}.
	\end{align*}
	It is easy to check that, for $\alpha > 1$, we have $-(2 \alpha)/(3 \alpha+1) > - \alpha/ (\alpha + 1)$. Therefore, if the first two relations above hold, then the right-hand side of (\ref{eq:goal 1}) tends to 0.
\end{proof}

\begin{proof} [Proof of Theorem \ref{theorem:optimal rate}]
	{\em 1}. If $\beta > (\alpha - 1)/(2 \alpha)$, then $ \ell  \lo 1 ( \delta) < \ell \lo 4 (\delta)$ for all $\delta > 0$, and consequently
	\begin{align*}
		m (\delta) = \max\{ \ell \lo 2 (\delta), \ell \lo 3 ( \delta), \ell \lo 4 ( \delta) \}.
	\end{align*}
	By computation, the intersection of $\ell \lo  2$ and $\ell \lo 4$ occurs at $\delta \lo {2,4} = {\alpha}/{(2\alpha \beta+ \alpha + 1)}$, and the intersection of $\ell \lo 3$ and $\ell \lo 4$ occurs at $\delta \lo {3,4} = {1}/{2}$. Moreover,   $\beta > ( \alpha - 1) / (2 \alpha)$ implies $\delta \lo {2,4} < 1/2 = \delta \lo {3,4}$. Hence the relative positions of the three lines $\ell \lo 2$, $\ell \lo 3$, $\ell \lo 4$ are  as depicted in Figure 1, left panel, and the minimum of $\max\{ \ell \lo 2 (\delta), \ell \lo 3 ( \delta), \ell \lo 4 ( \delta) \}$ is achieved at $\delta \loo {opt} = \delta \lo {2,4}$, with $m (\delta \loo {opt} ) = \ell \lo 2 (\delta \lo {2,4}) =  -\alpha (\beta \wedge 1)/\{2 \alpha (\beta \wedge 1) + \alpha + 1\}$.

	\noindent {\em 2}. If $\beta \le (\alpha - 1)/(2 \alpha)$, the $ \ell  \lo 1 ( \delta) \ge \ell \lo 4 (\delta)$ for all $\delta > 0$, and
	\begin{align*}
		m (\delta) = \max\{\ell \lo 1 (\delta),  \ell \lo 2 (\delta), \ell \lo 3 ( \delta) \}.
	\end{align*}
	The intersection of $\ell \lo  1$ and $\ell \lo 2$ occurs at  $\delta \lo {1,2}  = 1/2$, and the intersection of $\ell \lo 1$ and $\ell \lo 3$ occurs at $\delta \lo {1,3} =  {\alpha}/{(2\alpha \beta+ \alpha + 1)}$. Moreover,  $\beta < ( \alpha - 1) / (2 \alpha)$ implies $\delta \lo {1,3} > 1/2 = \delta \lo {1,2}$. Hence the relative positions of $\ell \lo 1$, $\ell  \lo 2$ and $\ell \lo 3$ are as shown in right plot of Figure \ref{fig:delta}, and the minimum of $\max\{\ell \lo 1 ( \delta ),  \ell \lo 2 (\delta), \ell \lo 3 ( \delta) \}$ is achieved at $\delta \loo {opt} = \delta \lo {1,2}=1/2$ with $m (\delta \loo {opt} ) = \ell \lo 2 (\delta \lo {1,2}) = {- \beta/2}$. 
\end{proof}

\begin{figure}[H]
	\vspace{-.1in}
	\centering
	\includegraphics[width=.4 \textwidth,height=.4\textwidth]{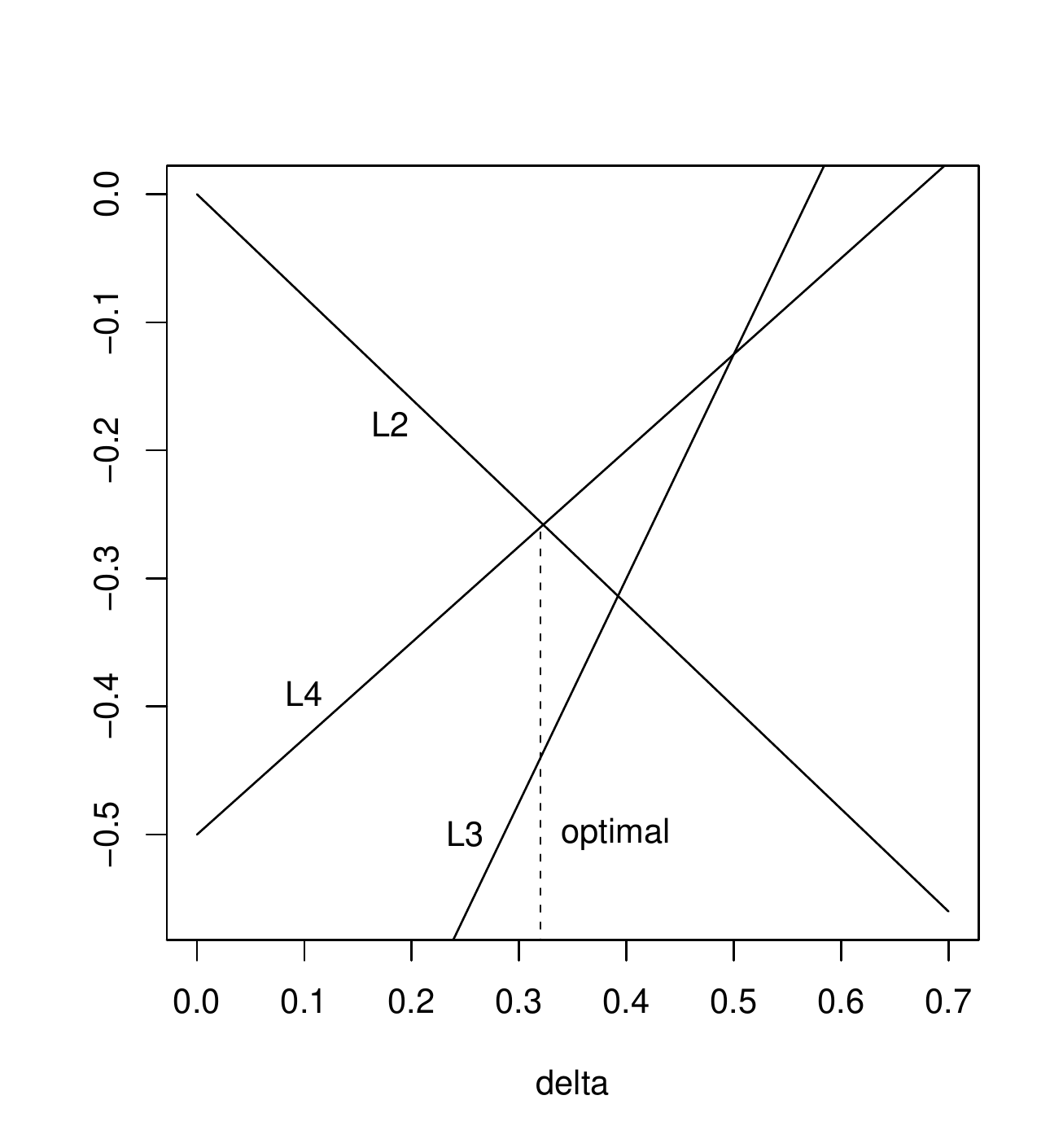}\hspace{.2in}\includegraphics[width=.4 \textwidth,height=.4\textwidth]{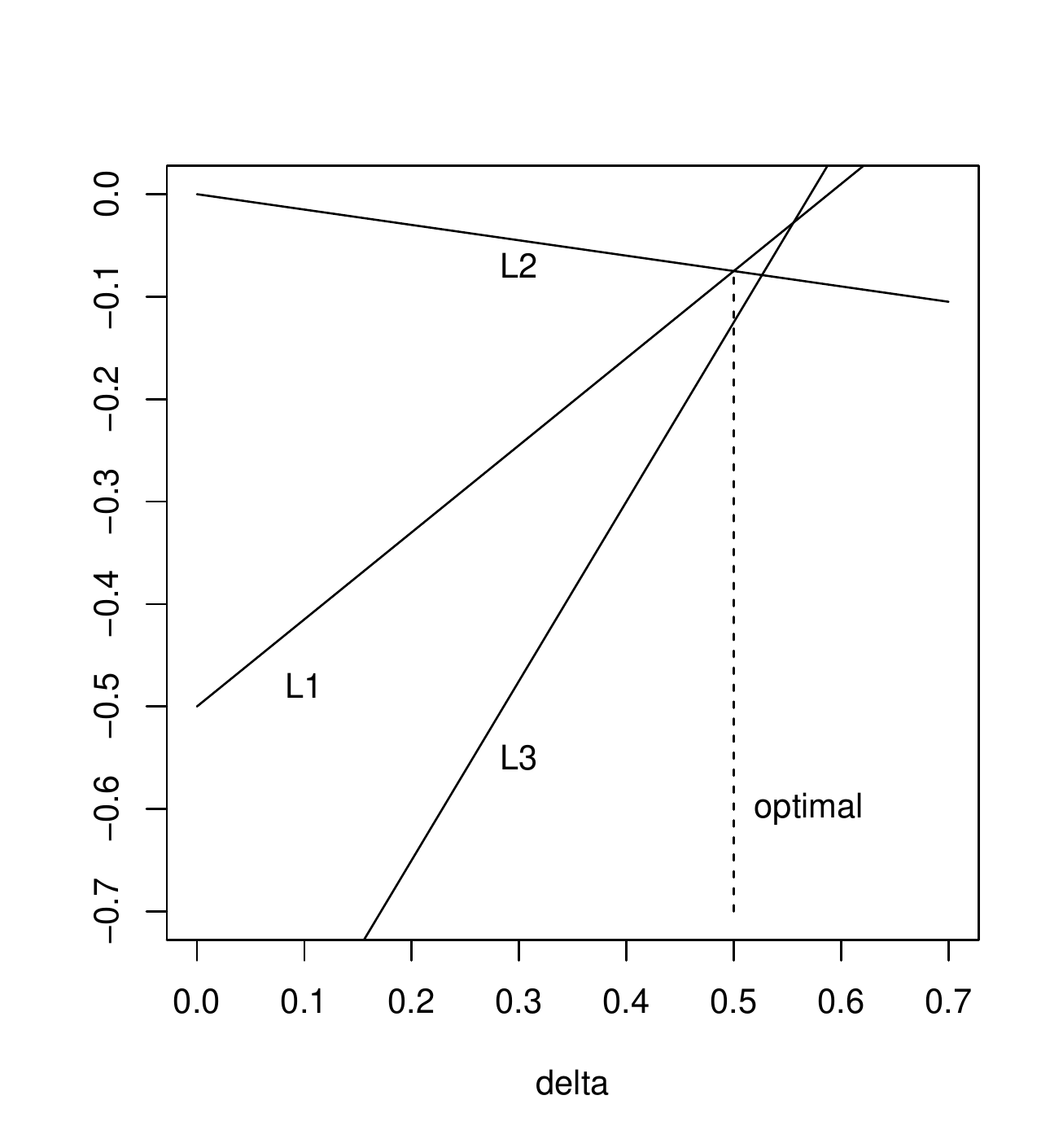}
	\vspace{-.2in}
	\caption{Optimal tuning parameter in two scenarios of $\beta$.  Left panel: L2, L3, L4 represent the lines $\ell \lo 2$, $\ell \lo 3$, $\ell \lo 4$  with $\beta > (\alpha-1)/(2 \alpha)$. Right panel: L1, L2, L3 represent the lines $\ell \lo 1$, $\ell \lo 2$, $\ell \lo 3$ with $\beta < (\alpha-1)/(2 \alpha)$.}
	\vspace{.3in}
	\label{fig:delta}
\end{figure}

\begin{proof}[Proof of Corollary \ref{corollary:regression estimate}]
	By (\ref{eq:true regression function}) and (\ref{eq:estimated regression function}), we have
	\begin{align}\label{eq:hat E expansion}
		\begin{split}
			\ali \widehat \E (Y | x \lo 0 )(t) - \E ( Y | x \lo 0)(t) \\
			\ali =  \langle \ka \lo T (\cdot, t), [\hat  B \hi * ( \ka (\cdot, x \lo 0) - \hat \mu \lo X)- B \lo 0  \hi * ( \ka (\cdot, x \lo 0) - \mu \lo X )]\rangle \lo {\ca H \lo Y} +  \hat  \mu \lo Y (t)  -    \mu \lo Y (t) \\
			\ali  =  \langle (\hat  B   - B \lo 0)\, \ka \lo T (\cdot, t),  \ka (\cdot, x \lo 0) - \mu \lo X \rangle \lo {\frak M \lo X}  + \langle (\hat  B  - B \lo 0 ) \, \ka \lo T (\cdot, t),  \mu \lo X - \hat \mu \lo X  \rangle \lo {\frak M \lo X} \\
			\ali \hspace{.2in}  + \langle  B \lo 0 \, \ka \lo T (\cdot, t),  \mu \lo X - \hat \mu \lo X \rangle \lo {\frak M \lo X}  +  [\hat  \mu \lo Y (t)  -    \mu \lo Y (t)].
		\end{split}
	\end{align}
	Hence
	\begin{align}\label{eq:hat E bound}
		\begin{split}
			\ali |\widehat \E (Y | x \lo 0 )(t) - \E ( Y | x \lo 0)(t) | \\
			\ali  \le  \|\hat  B   - B \lo 0 \| \loo {OP} \, \|\ka \lo T (\cdot, t) \| \lo {\ca H \lo Y} \, \| \ka (\cdot, x \lo 0) - \mu \lo X \| \lo {\frak M \lo X}
			+ \|\hat  B  - B \lo 0 \| \loo {OP} \,   \|\ka \lo T (\cdot, t) \| \lo {\ca H \lo Y} \,  \\
			\ali \hspace{.2in}  \|   \mu \lo X - \hat \mu \lo X  \| \lo {\frak M \lo X} +  \|  B \lo 0 \| \loo {OP} \, \|  \ka \lo T (\cdot, t)\| \lo {\ca H \lo Y} \, \| \mu \lo X - \hat \mu \lo X\| \lo {\frak M \lo X} +  | \hat  \mu \lo Y (t)  -    \mu \lo Y (t)|.
		\end{split}
	\end{align}
	Since $\hat \mu \lo Y$ and $\hat \mu \lo X$ are sample averages, by Chebychev's inequality,
	\begin{align*}
		\|   \mu \lo X - \hat \mu \lo X  \| \lo {\frak M \lo X} =O \lo P ( n \hi {-1/2}), \quad  \hat  \mu \lo Y (t)  -    \mu \lo Y (t)=O \lo P ( n \hi {-1/2}).
	\end{align*}
	Hence the right-hand side of (\ref{eq:hat E bound}) is dominated by $\|\hat  B   - B \lo 0 \| \loo {OP}$, which proves (\ref{eq:hat E rate}). The rest of the corollary is obvious. \end{proof}

\begin{proof}[Proof of Theorem \ref{theorem:clt}]
	{\em 1}. For convenience, let $g$ and $f$ denote the functions $\ka \lo T (\cdot, t)$ and $\ka (\cdot, x \lo 0) - \mu \lo X$. Then we can reexpress  $A \lo {n,3}$   as $n\inv \tsum \lo {i=1} \hi n Z \lo {ni}$ where
	\begin{align*}
		Z \lo {ni} =  \langle g, U \lo i \rangle \lo {\ca H \lo Y} \, \langle V\lo n  f, \ka (\cdot, X \lo i) - \mu \lo X \rangle \lo {\frak M \lo X}.
	\end{align*}
	Note that  $Z \lo {n1}, \ldots, Z \lo {nn}$ are i.i.d. random variables and, since $X \lo i  \indep U \lo i$,
	we have $\E  Z \lo {ni} =  0$. Hence
	\begin{align}\label{eq:E inner squared}
		\sigma \lo {n,3} \hi 2 = \E [\langle f, V \lo n \tilde   \Sigma \lo {XU} g  \rangle \lo {\frak M \lo X} \hi 2 ] = n \inv \E ( Z \lo {n1} \hi 2) = n  \inv \, \E [  \langle g, U \rangle \lo {\ca H \lo Y} \hi 2 \, \langle V \lo n f, \ka (\cdot, X) - \mu \lo X \rangle \lo {\frak M \lo X} \hi 2].
	\end{align}
	By $U \indep X$  and (\ref{eq:kappa expansion}), the right-hand side is
	\begin{align*}
		\ali n  \inv \, \E [  \langle g, U \rangle \lo {\ca H \lo Y} \hi 2 ] \, \E [ \langle V \lo n f, \ka (\cdot, X) - \mu \lo X \rangle \lo {\frak M \lo X} \hi 2] \\
		\ali \hspace{1in} = n  \inv \, \E [  U (t) \hi 2 ] \, \E [ \langle V \lo n f, \tsum \lo {j \in \nat} \zeta \lo j \varphi \lo j  \rangle \lo {\frak M \lo X} )\hi 2] \\
		\ali \hspace{1in}  = n  \inv \, \E [  U (t) \hi 2 ] \, \E [ (\tsum \lo {j \in \nat} \zeta \lo j (\lambda \lo j + \epsilon \lo n ) \inv  \langle  f,\varphi \lo j  \rangle \lo {\frak M \lo X}) \hi 2]. \end{align*}
	Since $\zeta \lo 1, \zeta \lo 2, \ldots$ are uncorrelated, we have
	\begin{align*}
		\E [ (\tsum \lo {j \in \nat} \zeta \lo j (\lambda \lo j + \epsilon \lo n ) \inv  \langle  f,\varphi \lo j  \rangle \lo {\frak M \lo X} \hi 2]
		=\ali \tsum \lo {j \in \nat} \E [ (\zeta \lo j \hi 2 (\lambda \lo j + \epsilon \lo n ) \hi {-2}  \langle  f,\varphi \lo j  \rangle \lo {\frak M \lo X} \hi 2] \\
		=\ali \tsum \lo {j \in \nat} \lambda \lo j  (\lambda \lo j + \epsilon \lo n ) \hi {-2}  \langle  f,\varphi \lo j  \rangle \lo {\frak M \lo X} \hi 2.
	\end{align*}
	Note that for any $j \geq 1$, $\varphi \lo j$ is a member of $\frak M \lo X \hi 0$, the effective domain of $\Sigma\lo{XX}$.
	Since $\langle f, \varphi \lo j \rangle \lo {\frak M \lo X} = \varphi \lo j (x \lo 0)- \E \varphi \lo j(X) = \varphi \lo j (x \lo 0)$, we have the desired equality in  part {\em 1}.
	
	\medskip
	
	\noindent {\em 2}. As argued in the proof of Corollary \ref{corollary:regression estimate}, the last three terms on the right-hand side of (\ref{eq:hat E expansion}) are all of the parametric order $O \lo P ( n \hi {-1/2})$ or smaller. Therefore we only need to consider the term
	\begin{align*}
		\langle (\hat  B \lo 0  - B \lo 0)\,g,  f \rangle \lo {\ca H \lo Y} = A \lo {n,1} + \cdots + A \lo {n,5}.
	\end{align*}
	By Assumption \ref{assumption:dominating}, $A \lo {n,3}$ is the dominating term, and so we only need to derive its asymptotic distribution.
	Since
	$A \lo {n,3}= n\inv \tsum \lo {i=1} \hi n Z \lo {ni}$, where   $\{ n \inv Z \lo {ni}: i = 1, \ldots, n, n \in \nat \}$ is a triangular array, we use Lyapounov's central limit theorem. Thus, for a $c > 0$, let
	\begin{align*}
		L \lo n (c ) = \sigma \lo {n,3} \hi  {-2-c} \tsum \lo {i=1} \hi n \, \E | n \inv Z \lo {ni}  | \hi {2 + c}.
	\end{align*}
	We need to  verify $L \lo n (c) \to 0$ as $n \to\infty$ for some $c > 0$. Take $c = 2$.
	Then
	\begin{align*}
		L \lo n (2)  = \sigma \lo {n,3}  \hi {-4}  \tsum \lo {i=1} \hi {n} \E ( | n\inv Z \lo {ni} | \hi 4 )= n \hi {-1} [ \E ( Z \lo {ni} \hi 2) ] \hi {-2}  \E (  Z \lo {ni}  \hi 4 ).
	\end{align*}
	By $U \indep X$,
	\begin{align*}
		\E (  Z \lo {ni} \hi 4 ) =  \E \langle U, g \rangle \lo {\ca H \lo Y} \hi 4 \, \E \langle V \lo n f, \ka (\cdot,X)- \mu \lo X  \rangle \lo {\frak M \lo X} \hi 4
		\asymp \E \langle V \lo n f, \ka (\cdot,X)- \mu \lo X  \rangle \lo {\frak M \lo X} \hi 4.
	\end{align*}
	Since the kernel $\ka$ is bounded, the  right-hand side is upper bounded by
	\begin{align*}
		\ali  \E (\langle V \lo n f, \ka (\cdot,X)- \mu \lo X  \rangle \lo {\frak M \lo X} \hi 2  \, \| V \lo n \| \loo {OP} \hi 2 \, \| f \| \lo {\frak M \lo X} \hi 2 \,\| \ka (\cdot,X)- \mu \lo X  \| \lo {\frak M \lo X} \hi 2 ) \\
		\ali \hspace{1in} \le  M \epsilon \lo n \hi {-2}   \E (\langle V \lo n f, \ka (\cdot,X)- \mu \lo X  \rangle \lo {\frak M \lo X} \hi 2   )
	\end{align*}
	for some $M > 0$.
	Also, recall that
	\begin{align}\label{eq:E Z squared}
		\E ( Z \lo {ni} \hi 2) = \E \langle U, g \rangle \lo {\ca H \lo Y} \hi 2 \, \E  ( \langle f, V \lo n ( \ka (\cdot, X)  - \mu \lo X ) \rangle \lo {\frak M \lo X} \hi 2 ).
	\end{align}
	Consequently,
	\begin{align*}
		L \lo n (2) = \frac{ O ( n \inv \epsilon \lo n \hi {-2} )}{ \E  ( \langle f, V \lo n ( \ka (\cdot, X)  - \mu \lo X ) \rangle \lo {\frak M \lo X} \hi 2 )}.
	\end{align*}
	By expansion (\ref{eq:kappa expansion}), the denominator above can be bounded below as follows:
	\begin{align*}
		\E  ( \langle f, V \lo n ( \ka (\cdot, X)  - \mu \lo X ) \rangle \lo {\frak M \lo X} \hi 2 )= \ali \tsum \lo {j \in \nat} \, ( \lambda \lo j + \epsilon \lo n ) \hi {-2} \lambda \lo j f \lo j \hi 2 \ge  ( \lambda \lo 1 + \epsilon \lo n ) \hi {-2} \lambda \lo 1 f \lo 1 \hi 2 \ge  ( \lambda \lo 1 + \epsilon \lo 1 ) \hi {-2} \lambda \lo 1 f \lo 1 \hi 2.
	\end{align*}
	Hence $L \lo n (2) = O(n \inv \epsilon \lo n \hi {-2})$. Since $\epsilon \lo n \succ n \hi {-1/2}$, we have $L \lo n (2) \to 0$.
\end{proof}

\begin{proof}[Proof of Lemma \ref{le-main}] 
	{\em 1}. By Proposition \ref{prop-condE}, recall that 
	\begin{align} \nonumber
		\ali \widehat{M}(X\lo{n + 1}) - \E(Y\lo{n+1} | X\lo{n + 1}) \\
		\label{eq-Mhatexpansion}
		& = \hat{\Sigma}\lo{YX}\hat{V}[\ka(\cdot, X\lo{n+1}) - \widehat{\mu}\lo X] - \Sigma\lo {YX}\Sigma\lo {XX}\mpinv [\ka(\cdot, X\lo{n+1}) - {\mu}\lo X] + \widehat{\mu}\lo Y - \mu\lo Y \\
		\nonumber
		& = (\hat{\Sigma}\lo {YX}\hat{V} - \Sigma\lo{YX}\Sigma\lo {XX}\mpinv)[\ka(\cdot, X\lo {n+1}) - {\mu}\lo X] + \hat{\Sigma}\lo {YX}\hat{V}(\mu\lo X - \widehat{\mu}\lo X) + \widehat{\mu}\lo Y - \mu\lo Y \\
		\nonumber
		& : = D\lo {1n} + D\lo {2n} + \widehat{\mu}\lo Y - \mu\lo Y 
	\end{align}
	in obvious correspondence. We study the second term first. By Lemma \ref{lemma:reg and res} (2.) and Lemma \ref{lemma:hat vs tilde}, we have
	\begin{align*}
		\|D\lo {2n}\| & = \|\hat{\Sigma}\lo {YX}(\hat{\Sigma}\lo {XX} + \epsilon\lo n I)\inv(\mu\lo X - \widehat{\mu}\lo X)\| \lo {\frak M \lo X} \\
		& = \|(\hat{\Sigma}\lo {UX} + B\lo 0\hi*\hat{\Sigma}\lo {XX}) (\hat{\Sigma}\lo {XX} + \epsilon\lo n I)\inv (\mu\lo X - \widehat{\mu}\lo X)\| \lo {\frak M \lo X}\\
		& = O\lo P(\epsilon\lo n\hi{-1} n\hi {-1/2}) [\|\hat{\Sigma}\lo {UX} - \tilde \Sigma \lo {UX}\|\loo{OP} + \|\tilde \Sigma \lo {UX}\| \loo {OP}] + O\lo P(n\hi {-1/2})\\
		& = O\lo P(\epsilon\lo n\hi{-1} n\hi {-1}) + O\lo P(n\hi {-1/2}) = O\lo P(n\hi {-1/2}). 
	\end{align*}
	The third relation holds since $\|\tilde \Sigma \lo {UX}\| \loo {OP} = O\lo P (n\hi {-1/2})$. For $ i = 1, \ldots, n + 1$, let $G\lo {i} := \ka(\cdot, X\lo{i}) - {\mu}\lo X\in \frak M \lo X \hi 0 $ ($P\lo X$- a.s.), and
	$\tilde{G}\lo i =  \ka(\cdot, X\lo {i}) - \hat{\mu}\lo X$. Then $D\lo {1n}$ can be expressed as:
	\begin{align*}
		D\lo {1n}\ali = \left(\hat{\Sigma}\lo{YX}\hat{V} - \Sigma\lo{YX}V\right) G\lo{n+1} \\
		\ali = B\lo 0\hi*(\hat{\Sigma}\lo{XX}\hat{V}-\Sigma\lo{XX}V\lo n)G\lo{n+1} + 
		B\lo 0\hi* ({\Sigma}\lo{XX}V\lo n - I) G\lo {n+1}   +  \frac{1}{n} \sum\lo {i = 1}\hi n U\lo i \langle \hat{V} \tilde{G}\lo i, G\lo {n+1} \rangle \lo {\frak M \lo X}   \\
		\ali - \frac{\bar{U} }{n} \sum\lo {i = 1}\hi n \langle \hat{V}\tilde{G\lo i}, G\lo {n+1} \rangle \lo {\frak M \lo X}.
	\end{align*}
	Since $\sum\lo {i = 1}\hi n \tilde{G}\lo i = 0$, the last term is 0. Denote the first three terms by $E\lo {1n}, E\lo {2n}$ and $E\lo {3n}$, respectively.

	We first deal with the bias term, $E\lo {2n}$. By the Karhunen-Lo\`eve theorem, we can rewrite $G \lo{n + 1} = \ka(\cdot, X\lo{n + 1}) - \mu\lo X$ as
	$\sum\lo {j = 1}\hi {\infty} \zeta\lo j \varphi\lo j$, where $\zeta\lo 1, \zeta\lo 2, \cdots$ are mean 0, uncorrelated random variables with variance $\lambda\lo 1, \lambda\lo 2, \cdots$.
	We now derive the form of $\E\|E\lo{2n}\|\lo{\ca H \lo Y}\hi 2$ to determine stochastic order of $E\lo{2n}$. Let $\{\psi\lo k: k \in \nat\}$ be an ONB of $\ca H \lo Y$. 
	Under Assumption \ref{assumption:finte moments and beta}({\em ii}), we have $\langle B\lo 0\hi* \varphi\lo j, \psi\lo k \rangle \lo {\ca H \lo Y} =  \langle\Sigma\lo{XX}\hi{\beta}\varphi\lo j, S\lo{XY}\psi\lo k\rangle \lo {\frak M \lo X} = \lambda\lo j\hi{\beta}s\lo {jk}$, where $s\lo {jk} = \langle \varphi\lo j, S\lo{XY}\psi\lo k\rangle \lo {\frak M \lo X}$. 
	Since $S\lo{XY}$ is a Hilbert-Schmidt operator, $\sum\lo j \sum\lo k s\lo {jk}\hi 2 < \infty$ holds, and hence
	\begin{align*}
		\E\|E\lo{2n}\|\lo{\ca H \lo Y}\hi 2 & = \E \left\{ \left\|B\lo 0\hi* \left[{\Sigma}\lo{XX}({\Sigma}\lo{XX} + \epsilon\lo n I)\hi {-1} - I\right] G\lo {n+1}\right\|\lo{\ca H \lo Y}\hi 2 \right\} \\
		& = \E \left\{ \left\|B\lo 0\hi* \sum\lo {j = 1}\hi {\infty} \frac{\epsilon\lo n\zeta\lo j}{\epsilon\lo n + \lambda\lo j} \varphi\lo j\right\|\lo {\ca H \lo Y}\hi 2 \right\} \\
		& = \E \left\{ \sum\lo {k = 1}\hi {\infty} \left\langle B\lo 0\hi* \sum\lo {j = 1}\hi{\infty} \frac{\epsilon\lo n\zeta\lo j}{\epsilon\lo n + \lambda\lo j} \varphi\lo j, \psi\lo k \right\rangle\lo{\ca H \lo Y}\hi 2 \right\} \\
		& = \epsilon\lo n\hi 2 \sum\lo {j = 1}\hi {\infty} \frac{\lambda\lo j\hi {2\beta + 1}}{(\epsilon\lo n + \lambda\lo j)\hi2} \sum\lo {k = 1}\hi {\infty} s\lo {jk}\hi 2.
	\end{align*}
	By straightforward calculation, we can verify that, when $\beta \geq 1/2$, $\sup\lo {j \geq 1} \frac{\lambda\lo j\hi {2\beta + 1}}{(\epsilon\lo n + \lambda\lo j)\hi 2} = O(1)$ as $n \rightarrow \infty$ regardless of the decaying rate of $\epsilon\lo n$. As a result, $E\lo {2n} = O\lo P(\epsilon\lo n)$.

	Next, we consider $E\lo{1n}$. Denote $\hat{\Sigma}\lo{XX} - \Sigma\lo{XX}$ by $\Delta$.
	Since $\hat{\Sigma}\lo{XX}\hat{V} -\Sigma\lo{XX}V\lo n = \epsilon\lo n\hat{V} \Delta V\lo n$, we have
	$$
	E\lo {1n} = \epsilon\lo n B\lo 0\hi* \hat{V} \Delta V\lo n G\lo {n+1} = (\epsilon\lo n S\lo{YX} \Sigma\lo{XX}\hi {\beta}\hat{V}\Delta) (V\lo n G\lo {n+1}) : = E\lo {11n} \times E\lo {12n}. 
	$$
	By Lemma \ref{lemma:fuku et al},  $\|\Delta\| = O\lo P(n\hi {-1/2})$.
	If we assume that $\beta \geq 1$, 
	\begin{align*}
		\|E\lo{11n}\| \ali \leq \epsilon\lo n \|S\lo{YX}\| \loo {OP}\times \|\Sigma\lo{XX}\hi{\beta - 1}\Sigma\lo{XX} (\hat{\Sigma}\lo{XX} + \epsilon\lo n I)\hi {-1}\| \loo {OP}\times \|\Delta\| \loo {OP} \\
		\ali = O\lo P(n\hi {-1/2}\epsilon\lo n) \times \|\Sigma\lo{XX}\hi{\beta - 1}\| \times \|[(\Sigma\lo{XX} - \hat{\Sigma}\lo{XX}) + \hat{\Sigma}\lo{XX}] (\hat{\Sigma}\lo{XX} + \epsilon\lo n I)\inv\|\loo {OP} \\
		\ali = O\lo P(n\hi {-1/2}\epsilon\lo n) O\lo P(n\hi {-1/2}\epsilon\lo n\hi{-1} + 1) = O\lo P(n\hi {-1/2}\epsilon\lo n),
	\end{align*}
	where the last relation holds since $\epsilon\lo n \succ n\hi {-1/2}$. By Lemma \ref{lemma:calculus}, we have 
	$$
	\E\left[\|E\lo{12n}\|\lo{\frak M \lo X}\hi 2\right]  = \sum\lo{j=1}\hi {\infty} \frac{\lambda\lo j}{(\epsilon\lo n + \lambda\lo j)\hi 2} = O(\epsilon\lo n\hi {-(\alpha+1)/\alpha}).
	$$ 
	It follows that 
	$E\lo{1n} = O\lo P(n\hi{-1/2}\epsilon\lo n) \times O\lo P(\epsilon\lo n\hi{-(\alpha+1)/(2\alpha)}) = O\lo P(n\hi {-1/2}\epsilon\lo n\hi {(\alpha-1)/(2\alpha)})$. 

	Lastly, we consider 
	$$
	E\lo {3n} = n\inv \sum\lo {i = 1}\hi n U\lo i \langle \hat{V} \tilde{G}\lo i, G\lo {n+1} \rangle \lo {\frak M \lo X} = \sum\lo {i=1}\hi n Z\lo{i,n} + \bar{U} \langle \hat{V} (\mu\lo X - \hat{\mu}\lo X), G\lo{n+1} \rangle \lo {\frak M \lo X},
	$$
	where $Z\lo{i,n} = \frac{1}{n} U\lo i  \langle \hat{V} G\lo i, G\lo {n+1} \rangle \lo {\frak M \lo X}$. 
	The remainder term satisfies that
	$$
	\|\bar{U} \langle \hat{V} (\mu\lo X - \hat{\mu}\lo X), G\lo {n+1} \rangle \lo {\frak M \lo X} \| \lo {\ca H \lo Y}  \leq \|\bar{U} \|\lo {\ca H \lo Y} \cdot \|\mu\lo X - \hat{\mu}\lo X \| \lo {\frak M \lo X}  \cdot \left\{\|V\lo n G\lo {n+1}\|\lo{\frak M \lo X} + \|(\hat{V} - V\lo n)G\lo{n+1}\|\lo{\frak M \lo X}\right\}. 
	$$
	Based our previous calculations of $\E(\|E\lo {12n}\|\hi 2)$, we have $\|V\lo n G\lo {n+1}\|\lo{\frak M \lo X} = O\lo P(\epsilon\lo n\hi {-(\alpha+1)/(2\alpha)})$. Note that
	$\hat{V} - V\lo n = -\hat{V}\Delta V\lo n$. Hence 
	\begin{align*}
		||(\hat{V} - V\lo n)G\lo {n+1}\|\lo{\frak M \lo X} \ali = \|\hat{V}\Delta V\lo n G\lo {n+1}\|\\
		\ali \leq \|\hat{V}\Delta\|  \|V\lo n G\lo {n+1}\| \\
		\ali = O\lo P(\epsilon\lo n\hi {-1}n\hi{-1/2}) O\lo P(\epsilon\lo n\hi {-(\alpha+1)/(2\alpha)}) \\
		\ali = o\lo P(\epsilon\lo n\hi {-(\alpha+1)/(2\alpha)}),
	\end{align*}
	where the last equality holds because $n\hi{-1/2} \prec \epsilon\lo n$. 
	It follows that the remainder is $O\lo P(n\hi {-1} \epsilon\lo n\hi {-(\alpha+1)/(2\alpha)})$. 

	\medskip
	
	\noindent {\em 2}. By definition, we have 
	\begin{align*}
		\|W\lo n\|\lo {\ca H \lo Y}\hi 2 \ali =  \frac{1}{n\hi 2} \sum\lo {i = 1}\hi n \|U\lo i\|\lo{\ca H \lo Y}\hi 2  \langle \hat{V} G\lo i, G\lo {n+1} \rangle \lo {\frak M \lo X}\hi 2 \\
		\ali + \frac{1}{n\hi 2} \sum\lo{i \neq i\hi {\prime}}\hi n \langle U\lo i, U\lo{i\hi {\prime}}\rangle\lo{\ca H \lo Y}  \langle \hat{V} G\lo i, G\lo {n+1} \rangle \lo {\frak M \lo X} \langle \hat{V} G\lo {i\hi{\prime}}, G\lo {n+1} \rangle \lo {\frak M \lo X}.
	\end{align*}
	By $(U \lo 1, \ldots, U \lo n) \indep (X\lo 1, \ldots, X\lo n)$ and $U \lo i \indep U \lo {i\hi{\prime}}$ if $i \neq i\hi{\prime}$, the expectation of the second term is 0. Moreover,  
	$\E\|W\lo n\|\lo {\ca H \lo Y}\hi 2 = n\inv\E[\|U\lo 1\|\lo{\ca H \lo Y}\hi 2 ] \E[\langle \hat{V} G\lo i, G\lo {n+1} \rangle \lo {\frak M \lo X}\hi 2] = n\inv{\sigma\lo {U}\hi 2}  \E[\langle \hat{V} G\lo i, G\lo {n+1} \rangle \lo {\frak M \lo X}\hi 2]$, where $\sigma\lo U\hi 2 = \tr(\Sigma\lo {UU}) < \infty$ since $\E\|Y\|\lo {\ca H \lo Y}\hi 2 < \infty$. 
	By the law of iterated expectations,
	\begin{align*}
		\E[\langle \hat{V} G \lo i, G\lo {n+1} \rangle \lo {\frak M \lo X}\hi 2] \ali = \E \left[\E\left(\langle \hat{V} G\lo i, G \lo {n+1} \rangle \lo {\frak M \lo X} \hi 2 \mid X\lo 1, \ldots, X\lo n\right) \right] \\
		\ali = \E\left[ \E\left\{\langle \hat{V} G\lo i, (G\lo{n+1} \otimes G\lo{n+1})\hat{V} G\lo i\rangle \lo {\frak M \lo X} \mid X \lo 1, \ldots, X \lo n \right\} \right] \\
		\ali = \E[\langle\hat{V}\Sigma\lo {XX}\hat{V} G\lo i, G\lo i \rangle\lo {\frak M \lo X}] \\
		\ali = \E[\tr (\hat{V}\Sigma\lo{XX}\hat{V} (G\lo i \otimes G\lo i))],
	\end{align*}
	where the third equality holds because $\E(G\lo{n + 1} \otimes G \lo {n + 1} | X \lo 1, \ldots, X \lo n) = \E(G\lo{n + 1} \otimes G \lo {n + 1}) = \Sigma\lo{XX}$, and the last equation holds because $\tr(A(g \otimes h)) = \langle Ag, h \rangle\lo {\ca H}$ for any $g, h \in \ca H$ and $A \in \ca B(\ca H)$.
	
	Let $\tilde{\Sigma}\lo{XX} = n \inv \sum\lo{i = 1}\hi n (\ka(\cdot, X\lo i) - \mu \lo X) \otimes (\ka(\cdot, X\lo i) - \mu\lo X)$. Since $\hat{V}\Sigma\lo{XX}\hat{V}(G\lo 1 \otimes G \lo 1), \ldots, \hat{V}\Sigma\lo{XX}\hat{V}(G\lo n \otimes G \lo n)$ have the same distribution, they have the same expectation. Hence
	\begin{align*}
		\E[\tr (\hat{V} \Sigma\lo {XX}\hat{V} (G\lo i \otimes G\lo i))] = \E[\tr (\hat{V}\Sigma\lo {XX}\hat{V} \tilde{\Sigma}\lo {XX})].
	\end{align*}
	By the properties of the trace of linear operators, we have
	\begin{align*}
		\E[\tr (\hat{V}\Sigma\lo {XX}\hat{V} \tilde{\Sigma}\lo {XX})] = \tr(\Sigma\lo {XX}\E(\hat{V} \tilde{\Sigma}\lo {XX}\hat{V})) \leq \tr(\Sigma\lo{XX}) \|\E(\hat{V} \tilde{\Sigma}\lo {XX}\hat{V})\|\loo{OP}.
	\end{align*}
	Rewriting $\hat{V} \tilde{\Sigma}\lo {XX}\hat{V}$ as $\hat{V} [(\tilde{\Sigma}\lo {XX} - \hat{\Sigma}\lo{XX}) +  \hat{\Sigma}\lo{XX}]\hat{V}$, we have
	\begin{align*}
		\|\E(\hat{V} \tilde{\Sigma}\lo {XX}\hat{V})\| \loo {OP} \ali  \leq \E(\|\hat{V} \tilde{\Sigma}\lo{XX}\hat{V}\|\loo{OP} ) \leq \E(\|\hat{V}\|\hi 2 \loo{OP} \|\tilde{\Sigma}\lo{XX} - \hat{\Sigma}\lo{XX}\|\loo{OP}) + \E(\|\hat{V}\|\loo{OP} \cdot \| \hat{\Sigma}\lo{XX}\hat{V}\|\loo{OP}) \\
		\ali \leq C\epsilon\lo n\hi{-2}n\hi{-1} + \epsilon\lo n\hi{-1} 
	\end{align*}
	for some constant $C > 0$. The last inequality holds since $\E( \|\tilde{\Sigma}\lo{XX} - \hat{\Sigma}\lo{XX}\|\loo{OP}) \leq C n\hi {-1}$ by the proof of Lemma 5 in \cite{fukumizu2007}. Since $\epsilon\lo n \succ n\hi {-1/2}$, we have $\epsilon\lo n\hi{-2}n\hi{-1} \prec \epsilon\lo n\hi {-1}$, proving Part {\em{2}}. 
\end{proof}

\begin{proof}[Proof of Theorem \ref{th-strongCLT}] 
	As argued following the proof of Lemma \ref{le-main}, 
	\begin{align*}
		\widehat{M}(X\lo {n + 1}) - \E(Y\lo {n+1} | X\lo {n + 1})  = F\lo {1n} + F\lo {2n} + W\lo n + R\lo n + O\lo P(n\hi{-1/2}).
	\end{align*}
	Since, by Assumption \ref{assumption: uniform dominating}, $W\lo n$ is the dominating term. We focus on the weak convergence of $W\lo n$ in $\ca H \lo Y$.  
	
	We first show the finite-dimensional convergence of $W \lo n$; that is,
	for any deterministic $y \in \ca H \lo Y$, 
	\begin{equation} \label{eq-finite}
		s\lo n\hi {-1} \langle W\lo n, y\rangle \lo {\ca H \lo Y} \stackrel{\ca D}\longrightarrow N(0, \sigma\hi 2\lo{U, y}),
	\end{equation}
	where $\sigma\lo{U,y} = \langle y, \Sigma \lo{UU} y \rangle \lo{\ca H \lo Y}$. 
	Let $\ca F\lo i$ denote the $\sigma$-algebra generated by $\{X\lo 1, U\lo 1, \ldots, X\lo i, U\lo i\}$ (or equivalently by $\{X\lo 1, Y\lo 1, \ldots, X\lo i, Y\lo i\}$). Let $H\lo i(y) : = \langle Z\lo {i, n}, y \rangle \lo {\ca H \lo Y}$. Obviously $\E[H\lo i(y)]$ is 0, and $H\lo i(y)$ is a martingale difference sequence with respect to the filtration $\ca F \lo i$. 
	To find its variance, we employ the law of iterated expectations:
	\begin{align*}
		\E \{H\lo i\hi 2(y) | \ca F \lo i\} & = \E\{\langle Z \lo {i, n}, y \rangle \lo {\ca H \lo Y} \hi 2 | \ca F \lo i \} \\
		& = \E \left\{ \left\langle \frac{1}{n} U\lo i \langle \hat{V}G \lo i, G\lo{n + 1} \rangle\lo{\frak M \lo X}, y \right\rangle \hi 2  | \ca F \lo i\right\} \\
		& = \frac{1}{n \hi 2} \E \left\{\langle \hat{V}G \lo i, G\lo{n + 1} \rangle\lo{\frak M \lo X} \hi 2 \langle U \lo i, y \rangle \hi 2 | \ca F \lo i\right\} \\
		& = \frac{\langle y, U\lo i  \rangle \hi 2 \lo {\ca H \lo Y}}{n\hi 2} \E \left\{ \langle \hat{V}G \lo i, G\lo{n + 1} \rangle\lo{\frak M \lo X} \langle \hat{V}G \lo i, G\lo{n + 1} \rangle\lo{\frak M \lo X} | \ca F \lo i \right\}
		\\
		& = \frac{\langle y, U\lo i  \rangle \hi 2 \lo {\ca H \lo Y}}{n\hi 2} \E \left\{ \langle \hat{V}G \lo i, (G\lo{n + 1} \otimes G \lo {n + 1}) \hat{V} G \lo i \rangle\lo{\frak M \lo X}  | \ca F \lo i \right\} \\
		& = \frac{\langle y, U\lo i  \rangle \hi 2 \lo {\ca H \lo Y}}{n\hi 2} \langle \hat{V}G \lo i, \Sigma\lo{XX} \hat{V} G \lo i \rangle\lo{\frak M \lo X} \\
		& = \frac{\langle y, U\lo i  \rangle \hi 2 \lo {\ca H \lo Y}}{n\hi 2} \|\Sigma\lo{XX}\hi{1/2} \hat{V} G\lo i\|\hi 2\lo{\frak M \lo X}. 
	\end{align*}
	Therefore, 
	\begin{align*}
		\E\{H\lo i\hi 2(y)\} & = \E\left[\E \{H\lo i\hi 2(y) | \ca F \lo i\}\right] \\
		& = \frac{1}{n\hi 2}\E \left[ \|\Sigma\lo{XX}\hi{1/2} \hat{V} G\lo i\|\hi 2        \E\left\{\langle y, U\lo i  \rangle\hi 2 \lo{\ca H \lo Y} | X\lo 1, \ldots, X\lo i\right\}        \right] \\
		& = \frac{\sigma \hi 2\lo {U,y}}{n\hi 2} \E\left( \|\Sigma\lo {XX}\hi{1/2} \hat{V} G\lo i\|\lo {\frak M \lo X}\hi 2 \right) = \frac{\sigma\lo {U,y}\hi 2 s\lo n\hi 2}{n},
	\end{align*}
	where, for the third equality, we used $U \lo i \indep (X\lo 1, \ldots, X\lo n)$.
	The convergence in \eqref{eq-finite} then follows from the central limit theorem for martingale difference arrays in \cite{mcleish1974}. 
	
	Next, we show that the sequence $s\lo n\hi {-1}W\lo n$ is asymptotically tight. By Lemma 1.8.1 of \cite{van1996weak}, it suffices to show that for any $\eta > 0$, 
	\begin{equation} \label{eq-tight}
		\limsup\lo {J \rightarrow \infty} \limsup\lo{n \rightarrow \infty} \Pr \left(\sum\lo {j > J}\langle s\lo n\hi {-1}W\lo n, \psi\lo j\rangle\hi 2\lo{\ca H \lo Y} > \eta\right) = 0, 
	\end{equation}
	where $\{\psi\lo j: j \in \nat \}$ is any ONB of $\ca H \lo Y$. For any $j \in \nat $, we have
	\begin{align*}
		\E\left(\langle s\lo n\hi {-1}W\lo n, \psi\lo j\rangle\hi 2\lo {\ca H \lo Y}\right) & = \E \left(\left\langle \frac{s\lo n \inv}{n} \sum\lo{i = 1}\hi n U\lo i \langle \hat{V}G \lo i, G\lo{n + 1} \rangle\lo{\frak M \lo X}, \psi\lo j \right\rangle \hi 2 \right) \\
		& = \frac{s\lo n \hi {-2}}{n \hi 2} \E \left\{\left(  \sum\lo{i = 1}\hi n \langle \hat{V}G \lo i, G\lo{n + 1} \rangle\lo{\frak M \lo X} \langle U \lo i, \psi\lo j\rangle \right) \hi 2\right\} \\
		& = \frac{s\lo n \hi {-2}}{n \hi 2} \sum\lo{i = 1}\hi n \sum\lo{i\pri = 1}\hi n 
		\E \left\{\langle \hat{V}G \lo i, G\lo{n + 1} \rangle\lo{\frak M \lo X} 
		\langle \hat{V}G \lo {i\pri}, G\lo{n + 1} \rangle\lo{\frak M \lo X} 
		\langle U \lo i, \psi\lo j\rangle \langle U \lo {i\pri}, \psi\lo j\rangle\right\}\\
		& = \frac{s\lo n \hi {-2}}{n \hi 2} \sum\lo{i = 1}\hi n \sum\lo{i\pri = 1}\hi n 
		\E\left\{\langle \hat{V}G \lo i, G\lo{n + 1} \rangle\lo{\frak M \lo X} 
		\langle \hat{V}G \lo {i\pri}, G\lo{n + 1} \rangle\lo{\frak M \lo X}\right\}
		\E\left\{\langle U \lo i, \psi\lo j\rangle \langle U \lo {i\pri}, \psi\lo j\rangle\right\}\\
		& = \frac{s\lo n \hi {-2}}{n \hi 2} \sum\lo{i = 1}\hi n
		\E\left\{\langle \hat{V}G \lo i, G\lo{n + 1} \rangle\lo{\frak M \lo X} 
		\hi 2\right\} \E\left\{\langle U \lo i, \psi\lo j\rangle \hi 2\right\}\\
		& = \E\left\{\langle U \lo i, \psi\lo j\rangle \hi 2\right\},
	\end{align*}
	where, for the fourth and fifth equalities, we used $(U \lo 1, \ldots, U \lo n) \indep (X\lo 1, \ldots, X\lo n)$, $U \lo i \indep U \lo {i\hi{\prime}}$ if $i \neq i\hi{\prime}$ and $\E(U \lo i) = 0$. Therefore, as $J \rightarrow \infty$, 
	$$
	\E\left(\sum\lo {j > J}\langle s\lo n\hi {-1}W\lo n, \psi\lo j\rangle\hi 2\lo {\ca H \lo Y}\right) = \E\sum\lo {j > J}\langle U, \psi\lo j\rangle\lo{\ca H \lo Y}\hi 2 \rightarrow 0,
	$$
	by the dominated convergence theorem, because the right-hand side is bounded by $\|U\|\lo{\ca H \lo Y}\hi 2$, which has a finite expectation. Thus the sequence $s\lo n\hi {-1}W\lo n$ is asymptotically tight. 
	
	Combining the above results, we obtain \eqref{eq-strongCLT}. 
\end{proof}
	
	
\bibliographystyle{apalike}
\bibliography{rkhs}

\begin{thebibliography}{}

\bibitem[Cai and Yuan, 2012]{cai2012}
Cai, T.~T. and Yuan, M. (2012).
\newblock Minimax and adaptive prediction for functional linear regression.
\newblock {\em Journal of the American Statistical Association},
  107(499):1201--1216.

\bibitem[Cardot et~al., 2003]{cardot2003}
Cardot, H., Ferraty, F., and Sarda, P. (2003).
\newblock Spline estimators for the functional linear model.
\newblock {\em Statistica Sinica}, 13:571--591.

\bibitem[Cardot et~al., 2007]{cardot2007}
Cardot, H., Mas, A., and Sarda, P. (2007).
\newblock {CLT} in functional linear regression models.
\newblock {\em Probability Theory and Related Fields}, 138:325--361.

\bibitem[Crambes and Mas, 2013]{crambes2013}
Crambes, C. and Mas, A. (2013).
\newblock Asymptotics of prediction in functional linear regression with
  functional outputs.
\newblock {\em Bernoulli}, 19(5B):2627--2651.

\bibitem[Cuesta-Albertos et~al., 2019]{cuesta2019}
Cuesta-Albertos, J.~A., Garc{\'\i}a-Portugu{\'e}s, E., Febrero-Bande, M., and
  Gonz{\'a}lez-Manteiga, W. (2019).
\newblock Goodness-of-fit tests for the functional linear model based on
  randomly projected empirical processes.
\newblock {\em The Annals of Statistics}, 47(1):439--467.

\bibitem[Ferraty and Vieu, 2006]{ferraty2006}
Ferraty, F. and Vieu, P. (2006).
\newblock {\em Nonparametric Functional Data Analysis: Theory and Practice}.
\newblock Springer, New York.

\bibitem[Friedman et~al., 2009]{friedman2001elements}
Friedman, J., Hastie, T., and Tibshirani, R. (2009).
\newblock {\em The Elements of Statistical Learning, 2nd edition}.
\newblock Springer, New York.

\bibitem[Fukumizu et~al., 2007]{fukumizu2007}
Fukumizu, K., Bach, F.~R., and Gretton, A. (2007).
\newblock Statistical consistency of kernel canonical correlation analysis.
\newblock {\em Journal of Machine Learning Research}, 8:361--383.

\bibitem[Fukumizu et~al., 2004]{fukumizu2004}
Fukumizu, K., Bach, F.~R., and Jordan, M.~I. (2004).
\newblock Dimensionality reduction for supervised learning with reproducing
  kernel {Hilbert} spaces.
\newblock {\em Journal of Machine Learning Research}, 5:73--99.

\bibitem[Fukumizu et~al., 2009]{fukumizu2009}
Fukumizu, K., Bach, F.~R., and Jordan, M.~I. (2009).
\newblock Kernel dimension reduction in regression.
\newblock {\em The Annals of Statistics}, 37(4):1871--1905.

\bibitem[Good et~al., 2013]{good2013}
Good, S.~A., Martin, M.~J., and Rayner, N.~A. (2013).
\newblock {EN4: Quality} controlled ocean temperature and salinity profiles and
  monthly objective analyses with uncertainty estimates.
\newblock {\em Journal of Geophysical Research: Oceans}, 118(12):6704--6716.

\bibitem[Hall and Horowitz, 2007]{hall2007}
Hall, P. and Horowitz, J.~L. (2007).
\newblock Methodology and convergence rates for functional linear regression.
\newblock {\em The Annals of Statistics}, 35(1):70--91.

\bibitem[Horv{\'a}th and Kokoszka, 2012]{horvath2012}
Horv{\'a}th, L. and Kokoszka, P. (2012).
\newblock {\em Inference for Functional Data with Applications}.
\newblock Springer, New York.

\bibitem[Hsing and Eubank, 2015]{hsing2015}
Hsing, T. and Eubank, R. (2015).
\newblock {\em Theoretical Foundations of Functional Data Analysis, with an
  Introduction to Linear Operators}.
\newblock John Wiley, Chichester.

\bibitem[Johnson and Horn, 1985]{johnson1985}
Johnson, C.~R. and Horn, R.~A. (1985).
\newblock {\em Matrix Analysis}.
\newblock Cambridge University Press, Cambridge.

\bibitem[Kokoszka and Reimherr, 2017]{kokoszka2017}
Kokoszka, P. and Reimherr, M. (2017).
\newblock {\em Introduction to Functional Data Analysis}.
\newblock CRC press, London.

\bibitem[Lee et~al., 2016]{lee2016variable}
Lee, K.-Y., Li, B., and Zhao, H. (2016).
\newblock Variable selection via additive conditional independence.
\newblock {\em Journal of the Royal Statistical Society. Series B (Statistical
  Methodology)}, 78(5):1037--1055.

\bibitem[Li, 2018]{li2018}
Li, B. (2018).
\newblock Linear operator-based statistical analysis: A useful paradigm for big
  data.
\newblock {\em Canadian Journal of Statistics}, 46(1):79--103.

\bibitem[Li and Song, 2017]{li2017}
Li, B. and Song, J. (2017).
\newblock Nonlinear sufficient dimension reduction for functional data.
\newblock {\em The Annals of Statistics}, 45(3):1059--1095.

\bibitem[Luo and Qi, 2017]{luo2017}
Luo, R. and Qi, X. (2017).
\newblock Function-on-function linear regression by signal compression.
\newblock {\em Journal of the American Statistical Association},
  112(518):690--705.

\bibitem[McLeish, 1974]{mcleish1974}
McLeish, D.~L. (1974).
\newblock Dependent central limit theorems and invariance principles.
\newblock {\em The Annals of Probability}, 2(4):620--628.

\bibitem[M{\"u}ller and Yao, 2008]{muller2008functional}
M{\"u}ller, H.-G. and Yao, F. (2008).
\newblock Functional additive models.
\newblock {\em Journal of the American Statistical Association},
  103(484):1534--1544.

\bibitem[Qi and Luo, 2019]{qi2019}
Qi, X. and Luo, R. (2019).
\newblock Nonlinear function-on-function additive model with multiple predictor
  curves.
\newblock {\em Statistica Sinica}, 29:719--739.

\bibitem[Ramsay and Silverman, 2005]{ramsay2005}
Ramsay, J.~O. and Silverman, B.~W. (2005).
\newblock {\em Functional Data Analysis 2nd edition}.
\newblock Springer-Verlag, New York.

\bibitem[Reimherr et~al., 2018]{reimherr2018}
Reimherr, M., Sriperumbudur, B., and Taoufik, B. (2018).
\newblock Optimal prediction for additive function-on-function regression.
\newblock {\em Electronic Journal of Statistics}, 12(2):4571--4601.

\bibitem[Shang and Cheng, 2015]{shang2015}
Shang, Z. and Cheng, G. (2015).
\newblock Nonparametric inference in generalized functional linear models.
\newblock {\em The Annals of Statistics}, 43(4):1742--1773.

\bibitem[Sun et~al., 2018]{sun2018optimal}
Sun, X., Du, P., Wang, X., and Ma, P. (2018).
\newblock Optimal penalized function-on-function regression under a reproducing
  kernel {Hilbert} space framework.
\newblock {\em Journal of the American Statistical Association},
  113(524):1601--1611.

\bibitem[Van Der~Vaart and Wellner, 1996]{van1996weak}
Van Der~Vaart, A.~W. and Wellner, J.~A. (1996).
\newblock {\em Weak Convergence and Empirical Processes}.
\newblock Springer, New York.

\bibitem[Yao et~al., 2005a]{yao2005a}
Yao, F., M{\"u}ller, H.-G., and Wang, J.-L. (2005a).
\newblock Functional data analysis for sparse longitudinal data.
\newblock {\em Journal of the American Statistical Association},
  100(470):577--590.

\bibitem[Yao et~al., 2005b]{yao2005b}
Yao, F., M{\"u}ller, H.-G., and Wang, J.-L. (2005b).
\newblock Functional linear regression analysis for longitudinal data.
\newblock {\em The Annals of Statistics}, 33(6):2873--2903.

\bibitem[Yuan and Cai, 2010]{yuan2010}
Yuan, M. and Cai, T.~T. (2010).
\newblock A reproducing kernel {Hilbert} space approach to functional linear
  regression.
\newblock {\em The Annals of Statistics}, 38(6):3412--3444.

\end{thebibliography}
	
\end{document}